\documentclass[lettersize,journal]{IEEEtran}
\usepackage[hyphens]{url}
\usepackage{hyperref}
\usepackage[hyphenbreaks]{breakurl}
\usepackage{epsfig}
\usepackage{amsfonts}
\usepackage{amssymb}
\usepackage{amstext}
\usepackage{amsmath}
\usepackage{xspace}
\usepackage{colordvi}
\usepackage{colordvi}
\usepackage[utf8]{inputenc}
\usepackage{mathrsfs}
\usepackage{amsmath,amsthm,bm}
\usepackage{xspace}
\usepackage{latexsym}
\usepackage{graphicx,epstopdf}
\usepackage{epsfig}
\usepackage{verbatim}
\usepackage{cases}
\usepackage{amsfonts,latexsym}
\usepackage{setspace}
\usepackage{tikz}
\usepackage{pgfplots}
\usepackage{footnote}
\usepackage{enumitem}
\usepackage{booktabs}
\usepackage{multirow}
\usepackage{tabularx}
\usepackage{makecell}
\usepackage{array}
\usepackage[square, comma, numbers]{natbib}

\usepackage{algorithm,algorithmic}

\makeatletter

\newcommand{\Rmnum}[1]{\expandafter\@slowromancap\romannumeral #1@}
\makeatother

\usepackage{array}
\usepackage{color}

\newtheorem{thm}{Theorem}

\newtheorem{ass}{Assumption}

\newcommand{\calm}{\mathcal{M}}     
\newcommand{\cals}{\mathcal{S}}     

\newcommand{\calt}{\mathcal{T}}

\newcommand{\cali}{\mathcal{I}}

\newcommand{\caly}{\mathcal{Y}}

\newcommand{\bt}{\boldsymbol{t}}
\newcommand{\bx}{\boldsymbol{x}}
\newcommand{\by}{\boldsymbol{y}}

\newcommand{\bp}{\boldsymbol{p}}

\newcommand{\opt}{\texttt{OPT}\xspace}
\newcommand{\alg}{\texttt{ALG}\xspace}

\newcommand{\CR}{\texttt{CR}\xspace}

\newcommand{\osac}{\texttt{OSARA}\xspace}

\newcommand{\csp}{\texttt{CSP}\xspace}
\newcommand{\sara}{\texttt{SARA}\xspace}
\newcommand{\opa}{\texttt{OPA}\xspace}
\newcommand{\fta}{\texttt{FTA}\xspace}

\newcommand{\microopt}{\texttt{RAA}\xspace}

\newcommand{\ie}{i.e., }
\newcommand{\cf}{cf., }

\newcommand{\fig}[1]{Fig.~\ref{#1}}

\newcommand{\sect}[1]{Section~\ref{#1}}

\usepackage{comment}
\usepackage{tabularx}
\usepackage{makecell}
\usepackage{threeparttable}
\usepackage{amsmath}

\usepackage{amsfonts}
\usepackage{amssymb}
\usepackage{mathrsfs}
\usepackage{multirow}
\usepackage{color, colortbl}
\usepackage{pgfplots}
\usepackage{tikz}
\usepackage{tikzscale}
\usepackage{soul}
\usepackage{enumitem}
\usepackage{etoolbox}
\usepackage{xcolor}
\usepackage{tabulary}
\usepackage[colorinlistoftodos]{todonotes}
\usepackage{algorithm, algorithmic}
\usepackage[normalem]{ulem}
\usepackage[compact]{titlesec}
\usepackage{todonotes}
\usepackage{subfig}
\usepackage{caption}

\pgfplotsset{
  every axis plot/.append style={line width=1.2pt},
}

\usetikzlibrary{arrows.meta, fit, backgrounds, shapes.geometric}
\tikzset{%
  >={Latex[width=2mm,length=2mm]},
         base/.style = {rectangle, rounded corners, draw=black,
                           minimum width=3cm, minimum height=1cm,
                           text centered, font=\rmfamily},
         terminal/.style = {base, circle, minimum size=1.5cm,font=\rmfamily}
}

\begin{document}

\bstctlcite{IEEEexample:BSTcontrol} %https://tex.stackexchange.com/questions/374074/how-do-i-use-max-num-names-before-forced-et-al
\title{Data-driven Online Slice Admission Control and Resource Allocation for 5G and Beyond Networks}

\author{\IEEEauthorblockN{ Muhammad Sulaiman\IEEEauthorrefmark{1}\thanks{\IEEEauthorrefmark{4}Equal contribution.}\IEEEauthorrefmark{4}, Bo Sun\IEEEauthorrefmark{1}\IEEEauthorrefmark{4}, Mohammad A. Salahuddin\IEEEauthorrefmark{1},\\Raouf Boutaba\IEEEauthorrefmark{1}, Aladdin Saleh\IEEEauthorrefmark{2} \\}
\IEEEauthorblockA{\{m4sulaim, b24sun, mohammad.salahuddin, rboutaba\}@uwaterloo.ca, aladdin.saleh@rci.rogers.com}\\
\IEEEauthorblockA{\IEEEauthorrefmark{1}University of Waterloo, 
\IEEEauthorrefmark{2}Rogers Communications Canada, Inc.} 
}

\maketitle
\IEEEpubidadjcol

\begin{abstract}

Virtualization in 5G and beyond networks allows the creation of virtual networks, or network slices, tailored to meet the requirements of various applications. However, this flexibility introduces several challenges for infrastructure providers (InPs) in slice admission control (AC) and resource allocation. To maximize revenue, InPs must decide in real-time whether to admit new slice requests (SRs) given slices' revenues, limited infrastructure resources, unknown relationship between resource allocation and Quality of Service (QoS), and the unpredictability of future SRs. To address these challenges, this paper introduces a novel data-driven framework for 5G slice admission control that offers a guaranteed upper bound on the competitive ratio, i.e., the ratio between the revenue obtained by an oracle solution and that of the online solution. The proposed framework leverages a pricing function to dynamically estimate resources' pseudo-prices that reflect resource scarcity. Such prices are further coupled with a resource allocation algorithm, which leverages a machine-learned slice model and employs a primal-dual algorithm to determine the minimum-cost resource allocation. The resource cost is then compared with the offered revenue to admit or reject a SR. To demonstrate the efficacy of our framework, we train the data-driven slice model using real traces collected from our 5G testbed. Our results show that our novel approach achieves up to $42\%$ improvement in the empirical competitive ratio, \ie ratio between the optimal and the online solution, compared to other benchmark algorithms. \end{abstract}
\begin{IEEEkeywords}
5G, Network Slicing, QoS, Admission Control, Optimization, Data-driven Algorithms, Online Algorithms
\end{IEEEkeywords}
\section{Introduction}
Network Function Virtualization (NFV) and Software-defined Networking (SDN) are enabling technologies to realize network slicing in 5G and beyond mobile networks. Network slicing allows the creation of isolated virtual networks, atop an underlying physical infrastructure, each tailored to meet the diverse requirements of distinct service modalities, such as enhanced Mobile Broadband (eMBB) and Ultra-Reliable Low-Latency Communications (URLLC). Anticipated future developments in 5G and beyond networks envision a commoditized landscape, where service providers (SPs) may procure network slices from infrastructure providers (InPs) to cater to specific market segments. However, constrained by limited resources, an InP may find itself unable to accommodate all slice requests (SRs) from myriad SPs. To manage its revenue, the InP could deploy a slice admission control (AC) mechanism that evaluates the resource demands of varying SRs against their offered revenues (i.e., values).

Slice AC in 5G and beyond networks presents several challenges. The first challenge is the \textit{absence of a well-defined slice model} that can precisely capture the relationship between the quality of service
(QoS) and the resource requirements for a given slice. This is attributed to  the service-level agreement (SLA) between SP and InP, which may specify QoS thresholds rather than concrete resource requirements. However, the QoS enjoyed by a slice depends on the resource distribution across various Virtual Network Functions (VNFs) within disparate network segments, and achieving equivalent QoS outcomes may be possible through different resource allocation strategies. Additionally, the SP's lack of visibility into the InP's proprietary VNF implementations exacerbates the problem as it precludes the stipulation of explicit resource requirements in advance. This is further complicated by the QoS being characterized by disparate key performance indicators (KPIs) across different types of slices. For instance, an URLLC slice may prioritize latency, while an eMBB slice may prioritize throughput. Therefore, an AC algorithm must possess the capability to deduce feasible resource allocation permutations that align with the specified QoS parameters prior to making an admission decision. Should admission be granted, the algorithm further needs to specify the optimal resource allocation permutation. This allocation should not only fulfill the QoS requirements for the newly admitted slice but also ensure resource availability for future SRs.

%In addition to the complex slice model, 
The second challenge of AC lies in the \textit{uncertainty about the slice information}. A prevalent assumption made in the current AC literature is that SR traffic specifications, i.e., the SR arrival-rate, their QoS requirements and the corresponding offered revenues, are either known in advance or can be characterized by a stationary distribution \cite{sulaiman2022coordinated, Ghina_2020, vanHuynh.2019, RR_2019}. These assumptions may not accurately reflect a realistic scenario where SRs are presented sequentially, necessitating immediate adjudication for each request independently of preceding or subsequent SRs. Additionally, the presence of multiple InPs offers SPs the freedom to redirect their SRs towards alternative InPs \cite{TNET-1}. Removing these assumptions precludes the solutions that model the AC problem as an offline problem, where the information of all SRs is known a priori. In addition, when SR traffic is non-stationary, Reinforcement Learning (RL)-based methods may struggle with non-convergence or prolonged convergence periods, due to their inherent assumption of a stationary environment \cite{rl-1, rl-2}.

The third challenge pertains to traditional AI-based AC algorithms, particularly those using RL \cite{sulaiman2022coordinated, Ghina_2020, vanHuynh.2019}, which often \textit{lack the robustness and explainability} needed for practical use. This is crucial, especially when denying SRs, as it hinders SPs from understanding or contesting decisions. Moreover, these algorithms do not offer performance guarantees and are susceptible to manipulations, such as an SP flooding the system with malicious SRs to monopolize resources, leading to inefficient network usage.

To tackle the aforementioned challenges, we propose a novel slice admission control (SAC) framework that integrates a data-driven slice model with Lagrangian decomposition and gradient-descent  for resource allocation. This resource allocation algorithm provides a near-optimal solution, which concurrently minimizes the resource usage cost and satisfies the QoS requirements. The proposed SAC framework achieves online SAC by using a pricing function that dynamically adjusts resource prices in real-time based on their scarcity. Consequently, only the SRs that offer higher revenue than their resource usage cost are accepted. The design of this pricing function ensures an upper limit on the competitive ratio, i.e., the ratio between the revenue obtained by an oracle solution and that of the online solution. In this paper, our main contributions are as follows: 

\begin{itemize}[leftmargin=*]
    \item \textbf{Joint resource allocation and slice admission control.} Our proposed SAC framework jointly optimizes resource allocation and online slice admission control. Instead of relying on a predetermined resource allocation, our approach dynamically determines resource allocation based on the SR's SLA and the current state of network resources. The admission control decision is then made using this real-time resource allocation.
    
   \item \textbf{Data-driven slice model and resource allocation.} Our slice modeling approach (from \cite{vnetrunner}) employs deep learning to model individual VNFs, and then composes the trained VNF models into an end-to-end (E2E) slice model. The E2E slice model is used to predict the QoS, and the QoS degradation is then calculated using the reparameterization trick. This allows to leverage back-propagation using existing ML-frameworks for optimizing resource allocation. In this work, we integrate this per-VNF slice model with our resource optimization algorithm (from \cite{microopt}). This gradient-based algorithm uses primal-dual optimization and achieves near-optimal resource allocation, while minimizing the resource usage cost and satisfying predefined QoS requirements.

    \item \textbf{Robust online algorithm for admission control.} We develop a novel online price-based algorithm for SAC. This algorithm sets a pseudo-price for each resource based on its utilization, estimates the total cost of serving the slice using the resource allocation algorithm, and admits the slice only if its value (i.e., offered revenue) is larger than the estimated cost. By carefully designing the resource prices, our proposed online algorithm is proven to attain a bounded competitive ratio, providing a theoretical guarantee on the worst-case performance under the competitive analysis framework \cite{borodin2005online}. 
    
    \item \textbf{Extensive experiments using real 5G testbed.} We train and validate the different components of the proposed framework using a full-fledged 5G testbed utilizing open-source components, including Open vSwitch (OvS) \cite{ovs}, srsRAN \cite{srsran}, and Open5GS \cite{open5gs}. The VNF and slice models are trained and evaluated using real traces from the testbed, and the resource allocation algorithm is compared to an optimal solution obtained through Gurobi. Unlike previous approaches that often assume fixed slice resource requirements, we integrate resource allocation as a subroutine within the online SAC algorithm. Our results show that that the proposed online SAC algorithm significantly outperforms benchmark algorithms.

    \item \textbf{Explainability.} Unlike the current state-of-the-art in slice admission control, our approach goes beyond mere slice acceptance or rejection decisions. It provides InP the ability to analyze the slice request admission decisions through cost/revenue analysis for each network resource. Additionally, the proposed algorithm can be used to derive the best achievable QoS given SR requirements and the offered revenue.
    
\end{itemize}

The remainder of the paper is organized as follows. In Section \ref{sec:bg_related}, we provide an overview of the related works. In Section \ref{sec:prob_statement}, we formally define the problem. Section \ref{sec:solution} describes the proposed solution, where we detail the resource allocation, and admission control algorithms. In Section \ref{sec:implementation}, we provide an exhaustive depiction of the testbed, elaborating on the technical specifications and configuration. Section \ref{sec:results} is dedicated to the discussion of results, where we expose the outcomes of our empirical evaluations. Finally, in Section \ref{sec:conclusion}, we conclude and instigate future research directions. Table \ref{tab:abbreviations} shows the frequently used abbreviations in this paper.

\begin{table}[!t]
\centering
\footnotesize
\caption{List of frequently used abbreviations}
\begin{tabular}{l l}
\toprule
\textbf{Abbr.} & \textbf{Meaning}\\
\midrule
CR & Competitive Ratio \\
CSP & Cost-minimization Single-slice Problem \\
E2E & End-to-end \\
ECR & Empirical Competitive Ratio \\
\fta & Fixed-threshold-based Algorithm \\
InP & Infrastructure Provider \\ 
\opa & Online Price-based Algorithm \\
\osac & Online Slice Admission Control and Resource Allocation\\
QoS & Quality of Service \\
\microopt & Resource Allocation Algorithm\\
RAN & Radio Access Network \\
SAC & Slice Admission Control \\
\sara & Slice Admission Control and Resource Allocation \\
SLA & Service-level Agreement \\
SP & Service Provider \\
SR & Slice Request \\
VNF & Virtual Network Function \\
\bottomrule
\end{tabular}
\label{tab:abbreviations}
\end{table}

\section{Related Works}\label{sec:bg_related}

The quintessential components of AC and resource allocation include: (i) a network model, (ii) a resource allocation algorithm, and (iii) an admission control algorithm. 
Specifically, the network model is used to estimate QoS based on SR specification and allocated resources. The resource allocation algorithm aims to satisfy SLA constraints while minimizing the resource usage cost. Finally, the AC algorithm is used to strategically accept or reject SRs based on their resource requirements and current network utilization, in order to maximize InP's revenue. In related works, these modules may not always be distinct. For instance, some studies assume known or easily derivable resource demands \cite{sulaiman2022coordinated,kasgari2018stochastic, salvat2018overbooking}, bypassing a separate resource allocation algorithm. On the other hand, RL for resource allocation implies that RL agents implicitly learn the slice model \cite{liu2021constraint, LiWeighted}.

\subsection{Network Modeling}

In many practical scenarios, testing various resource allocation policies directly on actual networks is infeasible. Consequently, network digital twins or network models provide a safer alternative for estimating network behavior. To this end, network simulators \cite{liu2021constraint, liu2022atlas} and Machine Learning (ML)-based estimators \cite{noms23, yang2022deepqueuenet, ferriol2023routenet} are frequently utilized. On one hand, traditional network simulators (e.g., ns-3) simulate the network at the packet-level and require substantial computation and time, limiting their use in online resource allocation \cite{yang2022deepqueuenet,ferriol2023routenet}. Additionally, these simulators often struggle to accurately replicate real-world conditions, particularly in wireless environments \cite{liu2022atlas}.

On the other hand, ML-driven approaches model the network using neural networks, which are trained using real-world or simulated network traffic traces. Once trained, these models offer near instantaneous inference of the network behavior. Regression-based network models (e.g., \cite{noms23}), Graph neural networks (GNNs) (e.g., \cite{ferriol2023routenet}), Bayesian neural networks (e.g., \cite{liu2022atlas}), and a combination of simulation and deep neural network (DNN) models (e.g., \cite{yang2022deepqueuenet}) have been used in this context. Our previous work in \cite{vnetrunner} evaluated network slice modeling across two dimensions: (i) VNF-level versus slice-level modeling, and (ii) packet-level versus flow-level modeling. The results demonstrated that VNF-level and flow-level modeling provides the optimal balance between accuracy and inference time.

\subsection{Resource Allocation}

Resource allocation and scheduling have been addressed using both ML and traditional optimization methods \cite{sulaiman2022coordinated, sciancalepore2019rl}. While conventional ML architectures, such as the encoder-decoder model \cite{bega2019deepcog}, have been utilized, Constrained Reinforcement Learning has emerged as a popular choice \cite{liu2021constraint,noms23,liu2021onslicing}. However, RL-based methods often suffer from extended convergence times \cite{ liu2021constraint,noms23}, rendering them impractical for online resource allocation. Additionally, these approaches necessitate retraining when actual online conditions differ from the training scenarios. \citet{liu2022atlas} introduced Bayesian optimization for SLA compliant resource allocation in network slices. However, similar to RL-based methods, their approach also suffers from lengthy convergence times in the order of several hours. Our previous work in \cite{microopt} introduced a gradient-based resource optimization approach, which utilizes a deep learning-based network slice model to achieve fast and near-optimal resource allocation.

Several studies have adopted traditional optimization methods for performance enhancement in multi-slice scenarios \cite{kasgari2018stochastic, salvat2018overbooking}. However, such approaches generally assume that resource requirements are known or can be easily inferred from SLA requirements. While this may be applicable to Physical Resource Block (PRB) allocation at a single base station, it does not extend to E2E slices, which require diverse resource types across different network segments, such as compute, and bandwidth resources for VNFs.

\subsection{Slice Admission Control}

% 5G Slice Admission Control (AC) involves the selective acceptance or rejection of slice requests (SRs) based on multiple criteria, including QoS requirements, duration of stay, and potential revenue. The primary goal of AC is to optimize the revenue of InPs. 
A number of recent works have addressed the challenge of 5G slice AC through various methodologies, predominantly utilizing trial-and-error-based techniques, \ie RL~\cite{sulaiman2022coordinated, vanHuynh.2019, RR_2019, TNSM_GNN, Roig.2019, bega2019machine} or the multi-armed bandit framework \cite{ONETS}.  \citet{Ghina_2020} proposed a traditional RL approach for managing 5G slice admission and congestion control. \citet{vanHuynh.2019} utilized a DRL-based slice AC and resource allocation. \citet{RR_2019} proposed a policy-based RL algorithm for slice AC in 5G C-RAN, with a focus on computing resource requirements at both remote and central sites based on latency needs. 
% \citet{bega2017} and its extension,
\citet{bega2019machine} used a multi-agent RL-based approach to maximize InP's revenue, where a separate agent predicts the revenue in case of accepting and rejecting SRs, respectively. Similarly, in our previous works \cite{sulaiman2022coordinated, noms22}, we proposed a multi-agent DRL-based approach to jointly optimize slice AC and VNF placement, and leveraged  GNNs \cite{TNSM_GNN} to accommodate for large scale and dynamic substrate network topologies.

Departing from RL-based methods, \citet{ONETS} introduced an online network slice brokering solution to maximize multiplexing gains, akin to a multi-armed bandit problem. Meanwhile, \citet{TNET-1} explored slice AC challenges, adding new constraints overlooked in previous literature. This includes scenarios with rational tenants and competitive InPs. The authors approached the multi-dimensional knapsack problem by dividing it into inter-slice admission and intra-slice quota allocation, proposing heuristic and auction mechanisms for these sub-problems.

However, the referenced works have several notable shortcomings. The primary limitation is the presumption of a stationary distribution of SRs \cite{sulaiman2022coordinated, Ghina_2020, vanHuynh.2019, RR_2019}. This assumption may not align with real-world scenarios where SR characteristics (e.g., arrival rates, resource demands, and offered revenues) may be non-stationary or even adversarial in nature. For instance, in dynamic market environments, SPs may alter their behaviors in response to the InP's current admission policies \cite{TNET-1}. RL-based methodologies traditionally rely on the premise of a stationary environment. Consequently, deviations in the SR distribution pose significant challenges, potentially hindering the convergence of RL-based approaches or resulting in high
regret.

Moreover, several of these studies presuppose the availability of explicit resource demand information with the SR \cite{Ghina_2020, vanHuynh.2019, RR_2019, TNET-1}. This assumption does not accurately reflect the operational modalities of 5G and beyond networks, where SPs are more likely to specify QoS thresholds rather than explicit resource requirements. The complexity of this issue is compounded by the variability in VNF implementations, which can result in different resource allocation combinations to achieve the same level of QoS.

\section{Problem Statement} \label{sec:prob_statement}

We address the joint Slice Admission control and Resource Allocation problem (\sara) for 5G and beyond network InPs. Consider an InP substrate network with $M$ general resources, where each resource $m\in \calm:=\{1,\dots,M\}$ represents either the bandwidth of a link or the computing resource of a node. Let $C_m$ denote the capacity of resource $m$.  
We assume a time-slotted system with a slot set $\calt:=\{1,\dots,T\}$, where a set of SRs 
% \red{[Terms are abbreviated multiple times in the manuscript, e.g., SR --- please check. Also, once abbreviated, use the abbreviation as much as possible]} 
$\cals :=\{1,\dots,S\}$ arrive sequentially. The information about each SR $i$ is represented by $I_i:= \{v_i, \calt_i, \boldsymbol{\xi}_i, \boldsymbol{\beta}_i\}$. $v_i$ is the value of SR, indicating price the SP is willing to pay for the request. $\calt_i:= \{t^a_i,\dots, t^d_i\}$ is the service period, where $t^a_i$ and $t^d_i$ are the arrival and departure slots of the SR, respectively. 
$\boldsymbol{\xi}_i$ is a slice feature vector that contains all necessary information for defining the SR's QoS (e.g., the QoS threshold $q^i_{\textit{thresh}}$ and the peak traffic  distribution $\Lambda_i$), while $\boldsymbol{\beta}_i$ is the acceptable QoS degradation threshold for the SR. Table \ref{tab:math_notations} shows the key notations used widely in the paper.

\begin{table}[!h]
\centering
\footnotesize
\caption{List of key notations}
\begin{tabular}{l l}
\toprule
\textbf{Notation} & \textbf{Description} \\
\midrule
$m\in \calm$ & Set of considered resources\\ 
$C_m$ & Capacity of resource $m$ \\
$\by_i$ & Resource allocation to SR $i$ \\
$\boldsymbol{\xi}_i$ & Feature vector of SR $i$, contains  $(q^i_{\textit{thresh}}, \boldsymbol{\beta}_i, v_i, \calt_i, \Lambda_i)$\\
$q_i(n_i, \by_i)$ & QoS distribution for slice $i$ with resources $\by_i$\\
$q^i_{\textit{thresh}}$ & QoS requirement for SR $i$ \\
$\boldsymbol{\beta}_i$ & QoS degradation threshold of slice $i$ \\
$v_i$ & Value (i.e., offered revenue) of SR $i$ \\
$\calt_i$ & Service period of SR $i$ \\
$n_i \sim \Lambda_i$ & Peak-traffic distribution for SR $i$\\
$R_m$ & Allocation upper limit for resource $m$ \\
$F_i(\by_i; \boldsymbol{\xi}_i, \calt_i)$ & QoS degradation function \\
$x_i$ & Admission control decision for SR $i$ \\
$w_{m,t}$ & Utilization of resource $m$ at time $t$ \\
$\phi_{m,t}(w)$ & Pricing function for resource $m$ at time $t$ \\
$\bp^{(i)}$ & Pseudo-price vector for SR $i$ \\
$c_i^*$ & Optimal cost for serving slice $i$ \\
$\sigma$ & Approximation factor for \csp solution\\
${\bt}_{\textit{in}} {\bt}_{\textit{out}}$ & Ingress and egress slice  traffic feature vector\\
$L, U$ & Lower and upper limits for slice value density\\
$V$ & Upper limit for variation of resource allocation\\
\bottomrule
\end{tabular}
\label{tab:math_notations}
\end{table}

\textbf{Decisions.} Upon the arrival of each SR $i\in\cals$, the problem is to immediately and irrevocably determine the admission control decision $x_i \in \{0,1\}$, i.e., whether to admit SR $i$, and a corresponding resource allocation $\by_i := \{y_{i,m}^t\}_{m\in\calm,t\in\calt}$. $y_{i,m}^t$ is the amount of resource $m$ allocated to SR $i$ at time slot $t$.
% and $y_{i,m}^t = 0, \forall \calt\setminus \calt_i$.
% If the resource allocation is fixed for each slice $i$ over its stay period, $y_{i,m}^t = y_{i,m}, \forall t\in\calt_i$.
The resource allocation of each slice $i$ for resource $m$ is constrained by $R_{m}$ to prohibit the slice from exhausting resource $m$. 
Given $x_i$, let $\caly_i(x_i)$ denote the feasible set of resource allocation $\by_i$. 
If SR $i$ is rejected (i.e., $x_i = 0$), no resource is allocated, and then $\caly_i(0):= \{\by_i: y_{i,m}^t = 0,\forall m\in\calm, t\in\calt\}$.
If SR $i$ is admitted (i.e., $x_i = 1$), a feasible resource allocation must guarantee the QoS requirement, therefore, $\caly_i(1):= \{\by_i: F_i(\by_i;\boldsymbol{\xi}_i,\calt_i) \le \boldsymbol{\beta}_i, R_m \ge y_{i,m}^t \ge 0, \forall m\in\calm, t\in\calt_i, y_{i,m}^t = 0, \forall m\in\calm, t\in \calt\setminus\calt_i\}$, where $F_i(\by_i;\boldsymbol{\xi}_i,\calt_i)$ denotes the QoS degradation during the service period $\calt_i$, given feature vector $\boldsymbol{\xi}_i$ and resource allocation $\by_i$. 

\textbf{Slice model and QoS requirement.}
% \bo{Need work to make connections.}
A network model determines the QoS metric (e.g., latency) for a given slice feature vector $\boldsymbol{\xi}_i$ and resource allocation $\by_i$.
%In this paper, we consider the key factor that affects the QoS is the peak-traffic of a slice, i.e., the number of users/s.
In this paper, we assume that the QoS is impacted by the peak-traffic of a slice, i.e., the number of users. Therefore, a network model $q_i(n_i,\by_i)$ gives the QoS performance when peak traffic is $n_i$ and resource allocation is $\by_i$.
The slice feature vector $\boldsymbol{\xi}_i$ includes a QoS threshold $q^i_{\textit{thresh}}$ and the peak-traffic distribution $\Lambda_i$, where $n_i \sim \Lambda_i$.  
 
Considering fair resource allocation and the same QoS experienced by each user on average, we can define the average QoS degradation for slice $i$ as
 \begin{align} \label{beta}         F_i(\by_i;\boldsymbol{\xi}_i,\calt_i) = \frac{\mathbb{E} \left[n_i \cdot{ {\mathbb{I}_{[q_{i}(n_i, \by_{i}) \leq q^i_{\textit{thresh}}]}}}\right]}{ {\mathbb{E} 
         \left[n_i\right]}}.
 \end{align}

It is required that the SLA of any admitted slice $i$ must be met, i.e., $F_i(\by_i;\boldsymbol{\xi}_i,\calt_i) \le \boldsymbol{\beta}_i$. 
Note that although this paper focuses on one specific way of defining the feature vector, network model and QoS constraint, our algorithmic framework can be customized for other network models and QoS constraints.
% (cf. Section~\ref{MicroOpt}).

\textbf{Offline problem.} The goal of \sara is to determine the admission control decision and resource allocation, such that the total value of admitted slices is maximized, the QoS requirements of all admitted slices are satisfied, and the capacities of all resources are respected.
Let $\cali :=\{v_i, \boldsymbol{\xi}_i, \calt_i, \boldsymbol{\beta}_i\}_{i\in\cals}$ denote an instance of the problem.
Given the information $\cali$ of all slices in advance, the offline problem can be written as:
\begin{subequations}
\label{p:sac}
\begin{align}
\label{eq:sac-obj}
    \max_{x_i, \by_{i}} \quad& \sum\nolimits_{i\in\cals} v_i x_i\\
    \label{eq:sac-capacity}
    {\rm s.t.}\quad & \sum\nolimits_{i\in\cals} y_{i,m}^t \le C_m, \forall m\in \calm, t\in\calt,\\ 
    \label{eq:sac-qos}
    &x_{i} \in \{0,1\}, \by_i \in \caly_i(x_i), \forall i\in\cals,
\end{align}
\end{subequations}
where the objective~\eqref{eq:sac-obj} maximizes the total value of admitted slices, constraint~\eqref{eq:sac-capacity} ensures that there are no capacity violations over all resources across the time horizon, and constraint~\eqref{eq:sac-qos} guarantees that feasible admission and resource allocation decisions can satisfy the QoS constraints.

\textbf{Online formulation.}
We also formulate an online version of \sara (\osac),
where the set of network slices arrive one by one. For each arrival, we must make admission decision and resource allocation without a priori knowledge of future slices.
Let $\alg(\cali)$ and $\opt(\cali)$, respectively, denote the total values obtained by an online algorithm and the offline algorithm under an instance $\cali$. The performance of the online algorithm is evaluated by its competitive ratio ($\CR$), i.e., $\CR = \max_{\cali \in \Omega} {\opt(\cali)}/{\alg(\cali)}$,
where $\Omega$ is the set of all possible instances. $\CR$ is a classic information-theoretic performance metric, which quantifies the performance of an online algorithm versus the offline algorithm in the worst-case scenario under the framework of competitive analysis~\cite{borodin2005online}. 
An algorithm with bounded \CR ensures robustness, and we aim to design an online algorithm that can minimize $\CR$.

\section{Online Algorithms for \osac} \label{sec:solution}
\begin{figure*}[h!]
 \centering
  \includegraphics[width=0.80\linewidth]{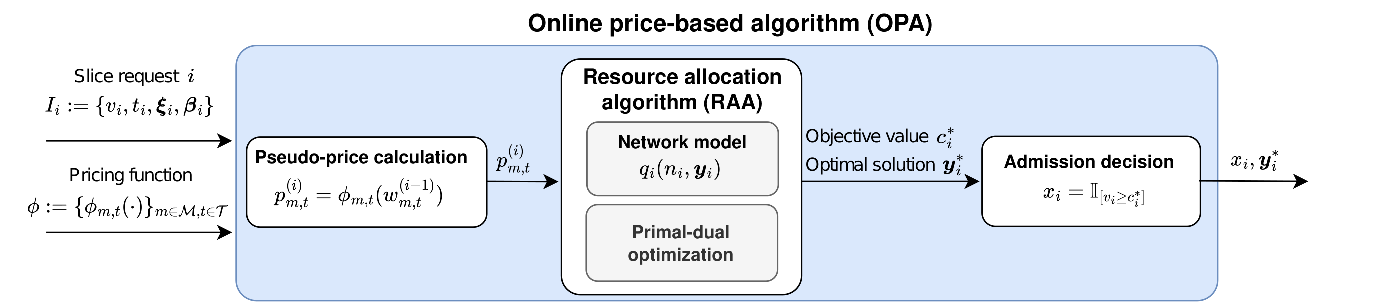}
    \caption{Overview of \opa for \osac}
    \label{fig:opa_raa}
\end{figure*}

Primarily, we face two main algorithmic challenges when designing algorithms for \osac: (i) \textit{unknown information about future requests}---without the knowledge of future slices, we must cautiously make admission decisions to balance the revenue from waiting for possible high-value slices with the risk that high-value slices may never show up, and (ii) \textit{network model without a precise formula}---due to the lack of an explicit formulation for QoS, we not only need to predict/construct the network model using observations or historical data, but also efficiently optimize resource allocation using the built model. 

To address the first challenge, we propose a price-based algorithm (\opa) that sets a price for using a unit of each resource at each time and admits a slice only when its value is larger than the cost of the resource consumption for serving the slice.
The prices are carefully designed to indicate the scarcity of resources and thus balance the immediate revenue and future opportunities.  

To handle the second challenge, we devise a data-driven resource allocation algorithm (\microopt). For this purpose, we construct a network model using deep learning and then optimize the resource allocation using a gradient-based method. \fig{fig:opa_raa} shows a high-level view of the proposed solution for \osac.
In the following, we provide the details of our algorithms that address these two challenges in Section~\ref{sec:alg-ac} and Section~\ref{MicroOpt}, respectively.

\subsection{Online Admission Control Algorithm}
\label{sec:alg-ac}

We propose $\opa$, an online price-based algorithm for \osac, in Algorithm~\ref{alg:ota}. 
The core idea is to estimate the cost of serving a slice based on the real-time resource utilization of the substrate network, and then admit the slice only when its value is larger than the estimated cost.
Towards this goal, we define a pricing function $\phi:=\{\phi_{m,t}\}_{m\in\calm,t\in\calt}$, where $\phi_{m,t}(w_{m,t})$ is a monotonically non-decreasing function that estimates the price of using resource $m$ at time $t$ when its utilization is $w_{m,t}$, where $w_{m,t} \in [0, C_m]$.

$\opa$ takes a pricing function $\phi$ as its input. 
For each slice $i$, the algorithm sets a pseudo-price vector as $\bp^{(i)} = \{p_{m,t}^{(i)}\}_{m\in\calm,t\in\calt}$, where $p_{m,t}^{(i)} = \phi_{m,t}(w_{m,t}^{(i-1)})$ and $w_{m,t}^{(i-1)}$ is the utilization of resource $m$ at time $t$ after processing the previous $i-1$ slices. 
Then \opa estimates the cost of admitting the slice $i$ by solving a single slice problem.

\textbf{Cost-minimization Single-slice Problem (\csp).}
For each slice $i\in\cals$, given the pseudo-price $\bp^{(i)}$ and the slice $i$'s information $\{\boldsymbol{\xi}_i,\boldsymbol{\beta}_i,\calt_i\}$, \opa solves a single slice problem $\csp(\bp^{(i)};\boldsymbol{\xi}_i,\boldsymbol{\beta}_i,\calt_i)$ to find the resource allocation that minimizes the cost of serving the slice. 
\begin{subequations}
\label{p:csp}
\begin{align}
\min_{\by_{i}} \quad& \sum\nolimits_{t\in\calt_i}\sum\nolimits_{m\in\calm} p_{m,t}^{(i)} y_{i,m}^t\\
\label{eq:csp}
    {\rm s.t.}\quad 
    & F_i(\by_i;\boldsymbol{\xi}_i,\calt_i) \le \boldsymbol{\beta}_i, \\
    & 0 \le y_{i,m}^t \le R_{m}, \forall m\in\calm, t\in\calt_i.
    % &y_{i,m} \ge 0, \forall i\in\cals, m\in\calm.
\end{align}
\end{subequations}
Let $\by_i^*$ and $c_i^*$ denote the optimal solution and objective value of \csp. Note that \csp is already a challenging problem since we cannot have a precise formula for $F_i(\by_i;\boldsymbol{\xi}_i,\calt_i)$ in constraint~\eqref{eq:csp}, and additionally constraint~\eqref{eq:csp} is usually not a convex constraint.  
Thus, \csp cannot be solved optimally most of the time. 
Suppose we have a $\sigma$-approximate algorithm that can obtain approximate  solution $\Tilde{\by}_i$ and objective value $\Tilde{c}_i$ of \csp such that $c_i^*\le \Tilde{c}_i \le  \sigma c_i^*$.
 
Then the admission control (in Line~\ref{alg:ac}) admits slice $i$ if the slice's value is larger than the scaled estimated cost $\Tilde{c}_i /\sigma$.
In addition, $\tilde{\by}_i$ is the corresponding resource allocation if the slice is admitted.

In this paper, we propose to solve \csp using a data-driven approach in Section~\ref{MicroOpt}.
Before proceeding to that section, we show that by carefully designing the pricing function $\phi$, $\opa$ can achieve a bounded competitive ratio, providing performance guarantees even under worst-case instances.

\begin{algorithm}[t!]
\caption{Online Price-based Algorithm ($\opa$)}
\begin{algorithmic}[1]
\STATE \textbf{Input:} pricing function $\phi:=\{\phi_{m,t}(\cdot)\}_{m\in\calm, t\in\calt}$;
\STATE \textbf{Initialization:} resource utilization $w_{m,t}^{(0)} = 0$, initial price $p_{m,t}^{(1)} = \phi_{m,t}(w_{m,t}^{(0)}), \forall m\in\calm, t\in\calt$;
\WHILE{slice $i$ arrives}
\STATE observe its value $v_i$, service period $\calt_i$, feature $\boldsymbol{\xi}_i$, and QoS requirement $\boldsymbol{\beta}_i$;
\STATE {solve the ancillary problem $\csp(\bp^{(i)};\boldsymbol{\xi}_i,\boldsymbol{\beta}_i,\calt_i)$ in problem~\eqref{p:csp}, and obtain a $\sigma$-approximate solution $\Tilde{\by}_i$, and the corresponding cost $\Tilde{c}_i$} that satisfies $\Tilde{c}_i \le \sigma c_i^*$; 
\IF{$v_i \ge \Tilde{c}_i/\sigma$}\label{alg:ac}
\STATE admit slice $i$ ($\bar{x}_i = 1$) and allocate resource $\bar{\by}_i = \Tilde{\by}_i$;
\ELSE
\STATE reject slice $i$ ($\bar{x}_i = 0$) and allocate resource $\bar{\by}_i = \boldsymbol{0}$;
\ENDIF
\STATE update the utilization $w_{m,t}^{(i)} = w_{m,t}^{(i-1)} + \bar{y}_{i,m}^t, \forall t, m$;
\STATE update the price $p_{m,t}^{(i+1)} = \phi_{m,t}(w_{m,t}^{(i)}), \forall t, m$;
\ENDWHILE
\STATE \textbf{Output:} admission and resource allocation $\{\bar{x}_i, \bar{\by}_i\}_{i\in\cals}$.
\end{algorithmic}
\label{alg:ota}
\end{algorithm}

\begin{ass}\label{ass1}
For each slice $i\in\cals$, its value $v_i$ and any feasible resource allocation $\by_i \in \caly_i(1)$ satisfy conditions:

(i) value density of each slice is bounded, i.e., 
\begin{align} \label{eq:ass-1}
    \frac{v_i}{T_i \sum_{m\in\calm}y_{i,m}^t} \in [L,U], \forall t\in \calt_i,
\end{align}
~~~~~~~where $T_i : = |\calt_i|$ is the length of the stay duration.

(ii) variation of resource allocation is upper bounded, i.e.,
\begin{align}\label{eq:ass-2}
    \frac{\sum_{m\in\calm}y_{i,m}^t}{\min_{m\in\calm: y_{i,m}^t>0}y_{i,m}^t} \le V, \forall t \in\calt_i.
\end{align}

(iii) resource allocation is small compared to capacity, i.e.,
\begin{align}\label{eq:ass-3}
    \max_{i\in\cals,t\in\calt_i} y_{i,m}^t \ll C_m, \forall m\in\calm.
\end{align}
\end{ass}
Condition~\eqref{eq:ass-1} requires that the value of each slice is proportional to its total resource consumption in each slot and stay duration, and the factor of the proportionality is uncertain but within bound $[L, U]$. Condition~\eqref{eq:ass-2} means that the resource allocation over different resources is relatively balanced, and the maximum variation is $V$. Condition~\eqref{eq:ass-3} assumes that the consumption of one slice for one resource is small compared to the capacity, which is reasonable in practice.    
We assume the parameters $L, U, V$ and the maximum stay duration $K = \max_{i\in\cals} T_i$, can be set by the InP beforehand. 

\begin{thm}\label{thm:osac}
Under Assumption~\ref{ass1}, given a $\sigma$-approximate algorithm for \emph{\csp},
\emph{\opa} is $\frac{(\sigma + 1)\alpha}{2}$-competitive for \emph{\osac} when the pricing function $\phi:=\{\phi_{m,t}\}_{m\in\calm, t\in\calt}$ is given by, $\forall m\in\calm, t\in\calt$,
\begin{align}\label{eq:threshold}
\phi_{m,t}(w) = L\left[\exp\left(\frac{\alpha w}{2 C_m} \right) - 1\right], w\in[0,C_m],
\end{align}
where $\alpha = 2\ln(\sigma UVK/L + 1)$.
\end{thm}
In the special case that \csp can be solved optimally (i.e., $\sigma = 1$), \opa will attain a competitive ratio of $2\ln(UVK/L+1)$. 
{Although it is challenging to derive the theoretical approximation ratio $\sigma$ for \csp, we show an approach to empirically estimate $\sigma$ in Section~\ref{sec:sigma}.
}

\begin{proof}[Proof of Theorem~\ref{thm:osac}]
We analyze the competitive performance of \opa based on the online primal-dual analysis approach~\cite{buchbinder2009design}.
By partially relaxing the capacity constraint~\eqref{eq:sac-capacity} using the dual variable $\boldsymbol{\lambda}:=\{\lambda_{m,t}\}_{m\in\calm,t\in\calt}$, the dual of problem~\eqref{p:sac} is
\begin{align*}
\min_{\boldsymbol{\lambda} \ge 0}\max_{\substack{x_i \in \{0,1\},\\ \by_{i}\in\caly_{i}(x_i)}} \quad \sum_{i\in\cals} v_i x_i + \sum_{t\in\calt}\sum_{m\in\calm} \lambda_{m,t} [C_m -  \sum_{i\in\cals} y_{i,m}^t].
\end{align*}
Equivalently, the dual problem can be presented as:
\begin{align*}
\min_{\boldsymbol{\lambda}\ge 0} \quad&  \sum_{t\in\calt}\sum_{m\in\calm} \lambda_{m,t} C_m + \sum_{i\in\cals} \hat{x}_i\left[v_i - \sum_{t\in\calt}\sum_{m\in\calm} \hat{y}_{i,m}^t \lambda_{m,t} \right], 
% {\rm s.t.} \quad&  v_i - \eta_i \le 0,
\end{align*}
where for a given $\boldsymbol{\lambda}$, $\hat{\by}_i:= \hat{\by}_i(\boldsymbol{\lambda}) = \{\hat{y}_{i,m}\}_{m\in\calm, t\in\calt}$ is the optimal solution of 
$$\min_{\by_i \in \caly_i(1)} \sum\nolimits_{t\in\calt_i}\sum\nolimits_{m\in\calm} \lambda_{m,t} y_{i,m}^t,$$
and $\hat{\bx}:=\{\hat{x}_i\}_{i\in\cals}$ is the optimal solution of 
$$\max_{x_i\in\{0,1\}} \sum_{i\in\cals} {x}_i\left[v_i - \sum_{t\in\calt}\sum_{m\in\calm} \hat{y}_{i,m}^t \lambda_{m,t} \right].$$ 
Thus, $\hat{\by}_i$ is the optimal solution of $\csp(\boldsymbol{\lambda};\boldsymbol{\xi}_i,\boldsymbol{\beta}_i,\calt_i)$, and
$\hat{x}_i = 1$ if $v_i \ge \sum_{t\in\calt}\sum_{m\in\calm} \hat{y}_{i,m}^t \lambda_{m,t}$ and $\hat{x}_i = 0$ otherwise.

The high-level idea of online primal-dual analysis is to construct a feasible dual solution of problem~\eqref{p:sac} based on the solution from the online algorithm \opa. Let $\texttt{Dual}(\cali)$ denote the dual objective evaluated at the feasible solution, then \opa is $c$-competitive if we can show that the following inequality holds: 
\begin{align}\label{eq:opd}
    c \cdot \alg(\cali) \ge  \texttt{Dual}(\cali) \ge  \opt(\cali). 
\end{align}

First, we show the second inequality, which holds based on weak duality when the solution $\{\bar{x}_i, \bar{\by}_i\}_{i\in\cals}$ of \opa is primal feasible and the constructed dual solution is dual feasible.
We note that the online decision of \opa satisfies constraint~\eqref{eq:sac-qos} directly. 
Therefore, we just need to show that no resource capacity constraints can be violated by the online decision of \opa. To see this, suppose resource $m'$ reaches the capacity, then for any follow-up slice that uses $m'$, the scaled estimated cost of admitting this slice is at least
\begin{align*}
\Tilde{c}_i/\sigma &\ge 
   \tilde{y}_{i,m'}^t \phi_{m',t}(C_{m'})/\sigma\\ &= \tilde{y}_{i,m'}^t VUK
   \ge \sum\nolimits_{m\in\calm}y_{i,m}^t U T_i \ge v_i, 
\end{align*}
where the first and the second inequalities hold due to conditions~\eqref{eq:ass-1} and~\eqref{eq:ass-2} in Assumption~\ref{ass1}.
Therefore, the follow-up slice will not be admitted by \opa and no capacity constraints will be violated.

Given the online decision of \opa, we construct a solution of the dual problem as: 
\begin{align}
    \bar{\lambda}_{m,t} = \phi_{m,t}(w^{(S)}_{m,t}), \forall m\in\calm, t\in\calt,
\end{align}
where $w^{(S)}_{m,t} = \sum_{i\in\cals} \bar{y}_{i,m}^t$ is the final utilization of resource $m$ at time $t$ when running \opa.
Clearly, $\bar{\lambda}_{m,t} \ge 0, \forall m\in\calm, t\in\calt$ and thus the dual solution is feasible.

Next, we show the first inequality in~\eqref{eq:opd}. Let $P_i$ and $D_i$ denote the primal and dual objective after processing the $i$-th slice using \opa.
The increment of the primal objective is 
\begin{align*}
    P_i - P_{i-1} = v_i \bar{x}_i,
\end{align*}
and the increment of the dual objective is 
\begin{align*}
    D_i - D_{i-1} &= \sum_{t\in\calt}\sum_{m\in\calm} [\phi_{m,t}(w^{(i)}_{m,t}) - \phi_{m,t}(w^{(i-1)}_{m,t})]C_m   \\
    &\quad\quad+\hat{x}_i [v_i - \sum_{t\in\calt}\sum_{m\in\calm} \hat{y}_{i,m}^t \phi_{m,t}(w^{(S)}_{m,t}) ].
\end{align*}

To relate the dual increment $D_i - D_{i-1}$ with online decision $\bar{x}_i$ and $ \bar{\by}_i$ of \opa, note that
\begin{subequations}
\begin{align}
  \hat{x}_i[v_i - \sum_{t\in\calt}&\sum_{m\in\calm} \hat{y}_{i,m}^t \phi(w^{(S)}_{m,t}) ] \\
  &\le \hat{x}_i [v_i - \sum_{t\in\calt}\sum_{m\in\calm} \hat{y}_{i,m}^t \phi(w^{(i-1)}_{m,t}) ]  \\
  &\le \hat{x}_i [v_i - \frac{1}{\sigma}\sum_{t\in\calt}\sum_{m\in\calm} \bar{y}_{i,m}^t \phi(w^{(i-1)}_{m,t}) ]\\
  &\le \bar{x}_i [v_i - \frac{1}{\sigma}\sum_{t\in\calt}\sum_{m\in\calm} \bar{y}_{i,m}^t \phi(w^{(i-1)}_{m,t}) ],
\end{align}    
\end{subequations}
where the first inequality holds since $\phi_{m,t}(\cdot)$ is a non-decreasing function.
The second inequality holds because $\hat{\by}_i$ and $\bar{\by}_i$ are the feasible solution and $\sigma$-approximate solution of $\csp(\{\bp^{(i)};\boldsymbol{\xi}_i,\boldsymbol{\beta}_i,\calt_i)$, respectively. Thus, $\sum_{t\in\calt}\sum_{m\in\calm} \hat{y}_{i,m}^t \phi(w^{(i-1)}_{m,t}) \ge c_i^* \ge \frac{1}{\sigma}\sum_{t\in\calt}\sum_{m\in\calm} \bar{y}_{i,m}^t \phi(w^{(i-1)}_{m,t})$. 
The last inequality holds because $\hat{x}_i = 1$ must give $\bar{x}_i = 1$.
Next we can consider the following two cases.

\noindent\textbf{Case I.} When $\bar{x}_i = 0$, we have $\bar{y}_{i,m}^t = 0, \forall m\in\calm, t\in\calt$,
and thus
\begin{align*}
D_i - D_{i-1} \le 0 = P_{i} - P_{i-1}.    
\end{align*}

\noindent\textbf{Case II.} When $\bar{x}_i = 1$, we have $\bar{\by}_i = \tilde{\by}_i$, and
$v_i \ge \frac{1}{\sigma}\sum_{t\in\calt_i}\sum_{m\in\calm} \bar{y}_{i,m}^t \phi_{m,t}(w^{(i-1)}_{m,t})$. Then we have
\begin{subequations}
\begin{align}
    D_i - D_{i-1} &\le \sum_{t\in\calt}\sum_{m\in\calm} [\phi_{m,t}(w^{(i)}_{m,t}) - \phi_{m,t}(w^{(i-1)}_{m,t})]C_m \nonumber\\
    &\quad\quad+  v_i - \frac{1}{\sigma}\sum_{t\in\calt}\sum_{m\in\calm} \bar{y}_{i,m}^t \phi_{m,t}(w^{(i-1)}_{m,t}) \\
    &\approx (\frac{\alpha}{2} - \frac{1}{\sigma})\sum_{t\in\calt_i}\sum_{m\in\calm}\bar{y}_{i,m}^t\phi_{m,t}(w_{m,t}^{(i-1)}) \nonumber\\
    \label{eq:ineq1}
    &\quad\quad\quad + \frac{\alpha}{2} \sum_{t\in\calt_i}\sum_{m\in\calm}\bar{y}_{i,m}^t L + v_i \\
    \label{eq:ineq2}
    &\le \frac{(\sigma+1)\alpha}{2} v_i  = \frac{(\sigma+1)\alpha}{2} (P_i - P_{i-1}).
\end{align}    
\end{subequations}
The equality~\eqref{eq:ineq1} holds since
\begin{subequations}
\begin{align*}
&\sum_{t\in\calt}\sum_{m\in\calm} [\phi_{m,t}(w^{(i)}_{m,t}) - \phi_{m,t}(w^{(i-1)}_{m,t})]C_m \\
&= C_m\sum_{t\in\calt_i}\sum_{m\in\calm} L \exp(\frac{\alpha w_{m,t}^{(i-1)}}{2 C_m}) \cdot [\exp(\frac{\alpha \bar{y}_{i,m}^t}{2 C_m} ) - 1] \\
&\approx C_m\sum_{t\in\calt_i}\sum_{m\in\calm} L \exp(\frac{\alpha w_{m,t}^{(i-1)}}{2 C_m} ) \cdot \frac{\alpha \bar{y}_{i,m}^t}{2 C_m} \\
& = \frac{\alpha}{2} \sum_{t\in\calt_i}\sum_{m\in\calm}\bar{y}_{i,m}^t[L \exp(\frac{\alpha w_{m,t}^{(i-1)}}{2 C_m} ) - L ] \nonumber\\
&\quad\quad\quad+  \frac{\alpha}{2}\sum_{t\in\calt_i}\sum_{m\in\calm}\bar{y}_{i,m}^t L,\\
&= \frac{\alpha}{2} \sum_{t\in\calt_i}\sum_{m\in\calm}\bar{y}_{i,m}^t\phi_{m,t}(w_{m,t}^{(i-1)}) +  \frac{\alpha}{2} \sum_{t\in\calt_i}\sum_{m\in\calm}\bar{y}_{i,m}^t L. \nonumber
\end{align*}
\end{subequations}
The inequality~\eqref{eq:ineq2} holds because: (i) $v_i \ge \frac{1}{\sigma}\sum_{t\in\calt_i}\sum_{m\in\calm} \bar{y}_{i,m}^t \phi(w^{(i-1)}_{m,t})$ from the decision rule in the online algorithm, and (ii) $\sum_{t\in\calt_i}\sum_{m\in\calm}  \bar{y}_{i,m}^t L \le \sum_{t\in\calt_i} \frac{v_i}{T_i} \le v_i$ from condition~\eqref{eq:ass-1} in Assumption~\ref{ass1}.

Thus, we have
\begin{align*}
    \texttt{Dual}(\cali) = D_S &= \sum\nolimits_{i\in\cals} [D_i - D_{i-1}] \\
    &\le \sum\nolimits_{i\in\cals} \frac{(\sigma+1)\alpha}{2} [P_i - P_{i-1}] \\
    &= \frac{(\sigma+1)\alpha}{2} P_S = \frac{(\sigma+1)\alpha}{2} \alg(\cali),
\end{align*}
which completes the proof.
\end{proof}

\subsection{Resource Allocation Algorithm} \label{MicroOpt}

\begin{figure*}[ht!]
\subfloat[Slice model]{  
\includegraphics[width=0.5\linewidth,trim={3.25em 0 2.5em 0},clip]{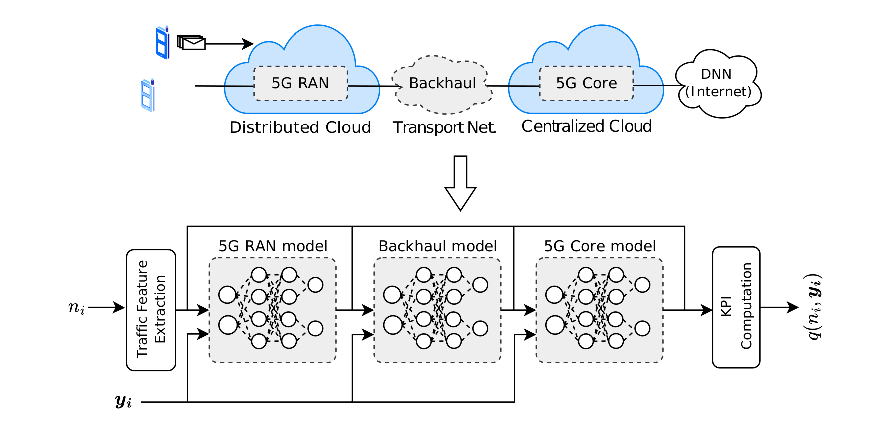}
    \label{fig:slice_model}}
    \subfloat[VNF model]{  
\includegraphics[width=0.5\linewidth,trim={4em 0 6em 0},clip]{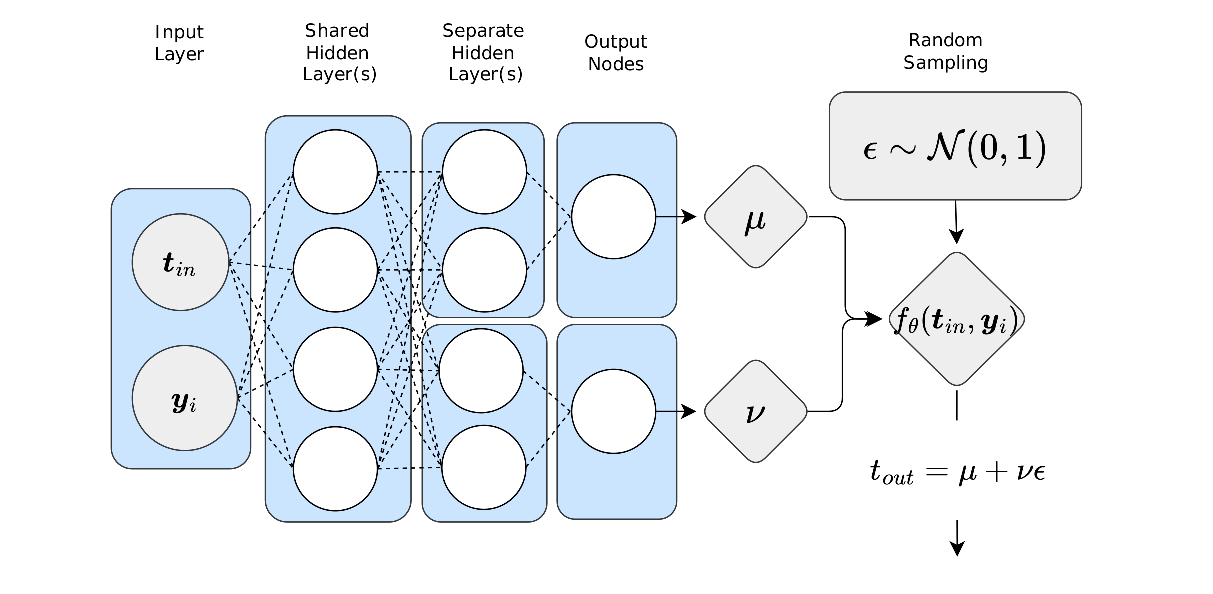}
    \label{fig:vnf_model}}
    \caption{Slice and VNF model architectures}
\end{figure*}

Given the definition of QoS degradation in \eqref{beta} and the unknown distribution of $q_{i}(n_i, \by_{i})$, \csp in \eqref{p:csp} pertains to a black-box continuous optimization problem. We solve this problem in two steps. In the first step, we use a data-driven framework that leverages the expressive power of neural networks to learn a \textit{slice model}. Given slice traffic $n_i$ and its resource allocation $\by_{i}$, the model predicts the QoS distribution $q_{i}(n_i, \by_{i})$, which can be used to compute the QoS degradation given in \eqref{beta}. In the second step, we tackle the \textit{constrained optimization} in \eqref{p:csp} by using a primal-dual optimization algorithm that capitalizes on the differentiability of the slice model.

\subsubsection{Slice Model} \label{sec:slice_model}

Slice modeling encompasses acquisition of the function $q_{i}(n_i, \by_{i})$ that captures the relationship among resource allocation, slice traffic, and QoS distribution. The QoS sampled from this distribution can be used to calculate the QoS degradation using~\eqref{beta}. Finally, the estimated QoS degradation is used for solving the constrained optimization problem in~\eqref{p:csp}. 

\fig{fig:slice_model} provides an overview of the proposed slice modeling framework. We start by modeling each VNF individually. The input and output of the VNF models is the traffic feature vector comprising flow-level traffic features, such as the mean and std. of the packet rate, packet size and inter-arrival time. Once trained, these individual models are then connected by passing the output of an upstream VNF as the input to the downstream VNF, creating an E2E model of a network slice. Finally, the distribution for the desired QoS metric can be computed using the egress traffic feature distribution. For instance, mean packet delay for the slice can be computed by adding the mean packet delay at each VNF. Similarly, throughput and jitter can be computed using egress packet rate and inter-arrival time, respectively. 

For the first VNF, the input feature vector can be computed using either publicly available traffic models \cite{etsi_traffic_profiles} or using private or public datasets. This network modeling approach assumes that the InP only offers pre-defined slice types to the SPs with known network models. However, the slices' attributes such as the service period $(\calt_i)$, peak traffic distribution $(\Lambda_i)$, QoS threshold $(q^i_{\textit{thresh}})$ and QoS degradation threshold $(\beta_i)$ can be customized by the SP.

\subsubsection{VNF Model}

In this section, we focus on the individual VNF models. As discussed in Section \ref{sec:bg_related}, different DNN-based architectures have been successfully demonstrated to model network behavior \cite{yang2022deepqueuenet, ferriol2023routenet}. We also use a DNN model to learn the input to output traffic feature vector relationship by employing a dataset that encompasses various resource allocations ($\by_{i}$) and input/output traffic feature vectors for each VNF. Unlike mathematical models (e.g., queuing models), a DNN-based model can easily handle heterogeneous types of resources and predict complex traffic features. Additionally, the complexity of this approach does not depend on the traffic volume, which is the case with packet-level simulators. 

We assume that the traffic feature vectors follows a normal distribution for the remainder of this paper, and the VNF model is designed to predict the parameters of this distribution. We choose the normal distribution as it proves to be sufficient for effectively modeling the data in our case (cf. numerical validation in 
Section~\ref{sec:exp-net-model}). However, it is important to note that the proposed VNF model can be extended to incorporate mixture density networks (MDNs), which have the ability to represent arbitrarily complex distributions \cite{mdn}. The architecture of the VNF model is shown in \fig{fig:vnf_model}. Let $f_{\theta}(\cdot)$ represent the function learned by the VNF model, parameterized by the neural network weights $\theta$. The inputs to the VNF model consist of the traffic feature vector ${\bt}_{\textit{in}}$ and resource allocation $\by_{i}$. These inputs are connected to a set of shared hidden layers, followed by separate hidden layers dedicated to each Gaussian distribution parameter. As a result, the model outputs the Gaussian distribution parameters $\mu$ and $\nu$ associated with the predicted output traffic feature vector ${\bt}_{\textit{out}}$. Under this distribution, the probability density function of $f_{\theta}({\bt}_{\textit{in}}, \by_{i})$ over the output traffic features ${\bt}_{\textit{out}}$ can be written as:
\begin{equation}
\begin{aligned}
\mathbb{P}({\bt}_{\textit{out}}\, |\, f_{\theta}({\bt}_{\textit{in}}, \by_{i})) = \frac{1}{\sqrt{2\pi}\nu} \exp\left(-\frac{\left({\bt}_{\textit{out}} - \mu\right)^2}{2\nu^2}\right).
\end{aligned}
\end{equation}
Finally, the loss for the model is computed as:
\begin{equation}\label{eq:lqos}
\begin{aligned}
L = -\frac{1}{B} \sum\nolimits_{j=1}^{B} \log \mathbb{P}({\bt}_{\textit{out},j}\, |\, f_{\theta}({\bt}_{\textit{in},j}, \by_{i,j})),
\end{aligned}
\end{equation}
where $B$ is the batch size and the subscript $j$ represents the $j$-th sample in the batch. This loss function calculates the negative log-likelihood of the ground-truth under the predicted normal distribution $\mathcal N(\mu, \nu)$ generated by the model for inputs $({\bt}_{\textit{in},j}, \by_{i,j})$. Once trained, this model can be used to sample the output traffic feature vector from the predicted distribution.

However, the drawback of na{\"i}vely sampling from the predicted distribution is that any subsequent optimization algorithm (\cf Section \ref{sec:opt}) that involves the sampled traffic feature vector would be limited to only numerical gradient calulation which is computationally expensive. To address this, we propose using the reparameterization trick, first introduced in \cite{vae_paper}. This technique, commonly employed in the ML literature, can also be used with other probability distributions, including MDNs \cite{mdn_reparam}. For this purpose, the random sampling is reformulated as follows: 
 \begin{equation}
\begin{aligned} \label{eq:qos_sample}
 t_{out} = \mu + \nu \epsilon,
\end{aligned}
\end{equation}
where $\epsilon$ is a random sample from a standard normal distribution $\mathcal N(0, 1)$ that does not depend on the inputs $({\bt}_{\textit{in},j}, \by_{i})$. The reparameterization trick not only allows for more efficient gradient calculation, it also allows the use of existing automatic differentiation frameworks (e.g., PyTorch \cite{torch}) for easy implementation. Finally, the computed gradients can be used in subsequent optimization to solve the \csp, as described in the following subsection.

 \subsubsection{Constrained Optimization} \label{sec:opt}

Once the slice model has been composed using the VNF models, we leverage gradient-descent along with primal-dual optimization to solve the constrained optimization problem in~\eqref{p:csp}. For this purpose, we start by converting the constrained problem into an unconstrained problem by using dual Lagrangian relaxation. The Lagrangian is defined as follows:
\begin{align}
&\mathcal{L}(\by_i, \lambda_s, \boldsymbol{\mu}) = \sum\nolimits_{t\in\calt_i} \sum\nolimits_{m\in\calm} p_{m,t}^{(i)} y_{i,m}^t +
\\
&\lambda_s \left( F_i(\by_i;\boldsymbol{\xi}_i,\calt_i) - \beta_{i} \right) + \sum\nolimits_{t\in\calt_i}\sum\nolimits_{m\in\calm} \mu_{m,t} \left( y_{i,m}^t - R_m \right),  \nonumber
\end{align}
where $\lambda_{s}$ denotes the Lagrange multiplier for the QoS degradation constraint, and $\boldsymbol{\mu} = \{\mu_{m,t}\}_{m\in\calm,t\in\calt_i}$ denotes the Lagrange multipliers for the resource limit constraint. Based on this formulation, the dual problem can be written as:
\begin{align}
\max_{\lambda_s \ge 0, \boldsymbol{\mu}\ge \boldsymbol{0}}  \min_{\by_i \ge \boldsymbol{0}} \quad \mathcal{L}(\by_i, \lambda_s, \boldsymbol{\mu}).  
\end{align}

The above dual problem can be solved iteratively using primal-dual updates with gradient-based methods \cite{boyd2004convex}, if it is differentiable with respect to both primal and dual variables. This is because gradient-based methods rely on the ability to compute the gradients of the objective function and the constraints with respect to relevant variables. 

Note that the computation of the QoS degradation $F_i(\by_i;\boldsymbol{\xi}_i,\calt_i)$ using \eqref{beta} involves an indicator function $\mathbb{I}_{[q_{i}(n_i, \by_{i}) \leq q^i_{\textit{thresh}}]}$, which is a piecewise constant and has a gradient of zero almost everywhere. This poses a challenge for gradient-based optimization algorithms that rely on gradient calculations for parameter updates \cite{cotter2019two}. 
To address this challenge, we introduce a surrogate QoS degradation function $\hat{F}_i(\by_i;\boldsymbol{\xi}_i,\calt_i)$ that replaces the indicator function in \eqref{beta} with a Sigmoid function $f(\rho*(q_{i}(n_i, \by_{i}) - q^i_{\textit{thresh}}))$, where  $\rho$ is a hyper-parameter that controls the sharpness of the curve. 
% \red{[Change the notation of $\sigma$ since it is already used for variance of the normal distribution]}. 
The Sigmoid function is a smooth and differentiable function and allows to use gradient-based optimization methods, while still approximating the behavior of the indicator function.

We denote the surrogate Lagrangian function, which incorporates the surrogate QoS degradation function, as 
$\mathcal{\hat{L}}(\by_i, \lambda_s, \boldsymbol{\mu})$. With this surrogate Lagrangian formulation, we can apply analytical gradient optimization techniques to optimize the resource allocation, while the solution's feasibility is ensured using the strict definition of QoS degradation. 

 \begin{algorithm}[t!]
 \caption{Resource Allocation Algorithm (\microopt)}
  \label{algo:res_alloc}
 \begin{algorithmic}[1]
 \renewcommand{\algorithmicrequire}{\textbf{Input:}}
 \renewcommand{\algorithmicensure}{\textbf{Output:}}
 \STATE \textbf{Input:} slice traffic ${\Lambda_i}$, network model $q_{i}(n_i, \by_{i})$, QoS threshold $q^i_{\textit{thresh}}$, QoS degradation threshold $\beta_i$, hyper-parameters $\tau_{1, max}, \tau_{2, max}, \eta_1, \eta_2, \eta_3, \epsilon_1, \epsilon_2$;
\STATE \textbf{Initialization:} $\lambda_s, \boldsymbol{\mu}, \textsc{LB}=0, \textsc{UB}=\infty, \tau_1=0, \tau_2=0$;
  \vspace{3pt}
  \WHILE{ $\frac{\textsc{UB} - \textsc{LB}}{\textsc{UB}} > \epsilon_1$ \OR $\tau_1 < \tau_{1, max}$ }
    \STATE $\by_i \leftarrow$ \texttt{Gridsearch} (${\Lambda_i}, q_{i}(n_i, \by_{i})$);
    \WHILE{ $|\nabla_y \mathcal{\hat{L}}| > \epsilon_2$ \OR $\tau_2 < \tau_{2, max}$ }
        \STATE $\by_i\leftarrow [\by_i - \eta_1 \nabla_y \mathcal{\hat{L}}]^
+$;
        \STATE $\tau_2 \leftarrow \tau_2 + 1$;
    \ENDWHILE
    \STATE $\lambda_s \leftarrow [\lambda_s + \eta_2(F_i(\by_i;\boldsymbol{\xi}_i,\calt_i) - \beta_i)]^+$;
    \STATE $\mu_{m,t} \leftarrow [\mu_{m,t} + \eta_3(y_{i,m}^t - R_m)]^+,\,\forall m, t$;
    \STATE $\textsc{LB} = \max\{\textsc{LB}, \mathcal{L}(\by_i, \lambda_s, \boldsymbol{\mu})\}$;
    \STATE $\textsc{UB} = \min\{\textsc{UB}, \sum\nolimits_{t\in\calt_i} \sum\nolimits_{m\in\calm} p_{m,t}^{(i)} y_{i,m}^t\}$;
    \STATE $\tau_1 \leftarrow \tau_1 + 1$;
  \ENDWHILE
 % \RETURN $\by_i$ 
\STATE \textbf{output:} resource allocation $\by_{i}$.
 \end{algorithmic} 
 \end{algorithm}

\textbf{Resource allocation algorithm (\microopt).} 
We propose a primal-dual optimization algorithm for resource allocation in
Algorithm~\ref{algo:res_alloc}. 
We refer to this algorithm as \microopt.
The algorithm takes as input the traffic ${\Lambda_i}$ and the network model $q_{i}(n_i, \by_{i})$. It also requires the QoS threshold $q^i_{\textit{thresh}}$, the QoS degradation threshold $\beta_i$ specific to each SR $i$, and several hyper-parameters to control the algorithm's behavior. These include parameters related to the stopping condition, such as $\tau_{1, max}$ and $\tau_{2, max}$, which define the maximum number of iterations for the outer and inner loops, respectively, and $\epsilon_1$ and $\epsilon_2$ that determine the desired level of convergence for the upper and lower bounds of the objective function. We also have learning rates, $\eta_1$, $\eta_2$, and $\eta_3$, for updating resource allocations and Lagrangian multipliers. Finally, the algorithm's output is the optimal resource allocation for each slice.

The algorithm is comprised of outer and inner loops. Within the inner loop, the resource allocation variables are updated using the gradient of the surrogate Lagrangian function ($\nabla_y \mathcal{\hat{L}}$). These updated variables are then projected into the non-negative domain denoted by the notation $[\cdot]^+$. We use coarse-grained \texttt{Gridsearch} to initialize the resource allocation variables. 
After updating the resource allocation variables, the algorithm updates the Lagrange multipliers inside the outer loop. QoS constraints multipliers, $\lambda_s$, are updated based on the QoS degradation values $F_i(\by_i;\boldsymbol{\xi}_i,\calt_i)$ and threshold $\beta_i$ for each slice. Similarly, resource constraints multipliers, $\mu_{m,t}$, are updated for each resource $m$ and slot $t$ by considering the difference between the allocated resources $y_{i,m}^t$ and the resource limit $R_m$.
At each point, the upper bound $\textsc{UB}$ is equal to the best feasible solution found so far, while the lower bound $\textsc{LB}$ is equal to the value of the Lagrangian function. Once the termination condition is met, the algorithm returns the resource allocation corresponding to the best $\textsc{LB}$.

\subsubsection{Estimating Upper Bound  on Approximation Ratio} \label{sec:sigma}

Since \csp involves the constraints from the slice model and is a non-convex problem, it is challenging to derive the optimal cost $c^*_i$ for each slice $i$ and a theoretical approximation ratio $\sigma$ for the solution provided by \microopt. 
However, in practice, we can utilize the information about SRs arrived during the online operation to continuously estimate an upper bound on the approximation ratio as a practical but sub-optimal solution.  

For each slice $i$, we can assume that the QoS degradation $F_i(\by_i;\boldsymbol{\xi}_i,\calt_i)$ is a monotonically non-increasing function of any resource \(y_{i,m}^t\), given that all other resources \(y_{i,m'}^t\) \((m' \neq m)\) remain fixed. This is because an increase in allocation for any resource will not degrade the slice's QoS.
Given this, we can derive the minimum allocation $y_{i,m}^{t,\min}$ of resource $m$ by setting all other resources $m'$ ($m'\not = m$) to the maximum possible value $R_{m'}$ and finding the minimum $y_{i,m}^{t,\min}$ that satisfies the QoS constraint, i.e., 
\begin{align}
    y_{i,m}^{t,\text{min}} = \min \{ y_{i,m}^t \,|\, F_i(R_1, \dots, y_{i,m}^t, 
    \dots, R_M;\boldsymbol{\xi}_i,\calt_i) \leq \beta\} \nonumber.
\end{align}
% where $R_{m'}$ is the maximum feasible value of $y_{i,m'}$. 
We can use a binary search to compute $y_{i,m}^{t,\text{min}}$ efficiently.

Based on the minimum allocation $y_{i,m}^{t,\text{min}}$, we can derive a lower bound for the optimal cost $c_i^*$ as $c^\text{lb}_i = \sum_{t\in\calt_i}\sum_{m \in \mathcal{M}} p_{m} y_{i,m}^{t,\text{min}}$,
% \begin{align*}
%     c^\text{lb}_i = \sum_{t\in\calt_i}\sum_{m \in \mathcal{M}} p_{m} y_{i,m}^{t,\text{min}}.
% \end{align*}
and then estimate an upper bound on the approximation ratio of \microopt by
\begin{align*}
    \sigma^\text{ub}_i = \frac{c_i}{c^\text{lb}_i} \ge \frac{c_i}{c_i^*},
\end{align*}
where $c_i$ is the cost of the solution from \microopt. We propose to use the max $\sigma^{\text{ub}} = \max[\sigma^\text{ub}_i]$ of the upper bound over SR distributions to estimate the approximation ratio $\sigma$. In Fig.~\ref{fig:sigma} of Section~\ref{sec:exp-resource-alloation}, we empirically validate that the estimation $\sigma^{\text{ub}}$ is close to the actual approximation ratio.
% the minimum resource allocation vector $\by_i^\text{min}$ can be constructed as:
% \begin{align}
%     \by_i^\text{min} = [y_{i,1}^\text{min}, y_{i,2}^\text{min}, \dots, y_{i,M}^\text{min}].
% \end{align}

% Let the total cost of allocating $\by_i^\text{min}$ for slice $i$ be denoted by $c_\text{lb}$:
% \begin{align}
%     c_\text{lb} = \sum_{m \in \mathcal{M}} p_{m} y_{i,m}^\text{min}.
% \end{align}

% Next, let $c^*$ denote the optimal cost obtained by solving the resource allocation problem exactly and $c$ denote the cost obtained using \microopt. By definition, we have:
% \begin{align}
%     c_\text{lb} \leq c^* \leq c.
% \end{align}

% The approximation ratio $\sigma$ of the \microopt is expressed as:
% \begin{align}
%     \sigma = \frac{c}{c^*}.
% \end{align}

% Using the inequality $c_\text{lb} \leq c^*$, we can upper bound $\sigma$ by:
% \begin{align}
%     \sigma \leq \sigma_\text{ub}, \quad \text{where } \sigma_\text{ub} = \frac{c}{c_\text{lb}}.
% \end{align}

% Thus, $\sigma_\text{ub}$ provides an upper bound on $\sigma$, ensuring that $\sigma_\text{ub} \geq \sigma$. The approach described provides a systematic way to determine $\sigma_\text{ub}$ for the \microopt and only $O(|M|log(\sum_{m \in \mathcal{M}}C_m/c_m))$ queries to the slice model, where $c_m$ is the smallest possible change to resource $m$ that the InP considers relevant.

\section{Implementation}\label{sec:implementation}
In this section, we describe the implementation of our network slicing testbed, shown if \fig{fig:testbed} . We use this testbed to collect our dataset, and then construct a data-driven slice model based on the dataset. The testbed deployment instruction are publicly available in \cite{testbed_repo}.

\subsection{Testbed Infrastructure} \label{sec:testbed}

The testbed is deployed on a three-node Kubernetes cluster. A high-performance physical machine with 32 CPU cores and 32 GB of RAM is dedicated to hosting the RAN. The transport network and core are deployed on Intel NUC PCs, each configured with 8 CPU cores and 16 GB of RAM. All three nodes are connected through a 1~Gbps NETGEAR switch.

\begin{figure}[ht!]
    \centering

    \includegraphics[width=\linewidth]{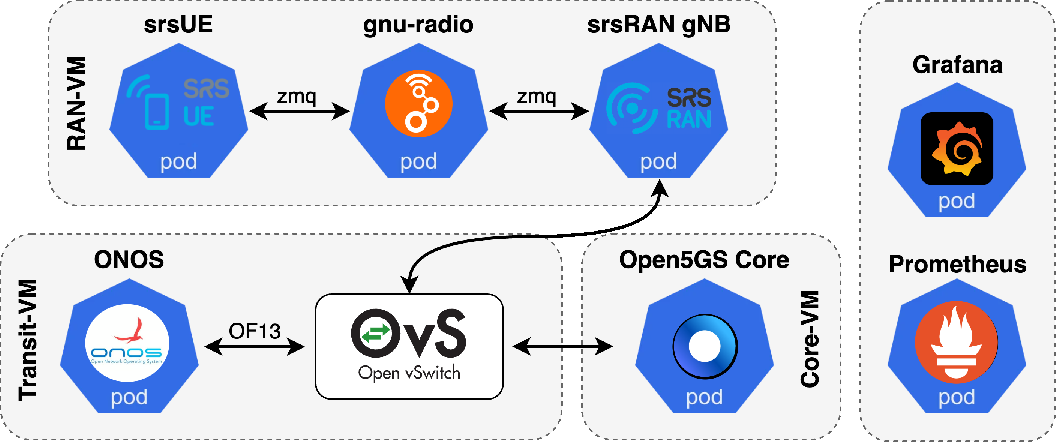}
    \caption{Overview of our 5G testbed}
    \label{fig:testbed}
\end{figure}

\subsection{5G Network Implementation}

\noindent\textbf{RAN.} The 5G RAN is implemented using the srsRAN project \cite{srsran}, an open-source software that provides a 3GPP Release 17 (R17) compliant gNB. User Equipments (UEs) are emulated with srsUE \cite{srsUE}. Virtual radios from srsRAN are utilized instead of physical radios to enable communication between the gNB and UEs. Additionally, GNU Radio Companion is used to handle the uplink and downlink signals.

\smallskip
\noindent \textbf{Core.} The 5G core network is implemented using Open5GS \cite{open5gs}, an open-source 3GPP Release 17 (R17) compliant implementation. Core functions, including the Access and Mobility Management Function (AMF), Session Management Function (SMF), User Plane Function (UPF), and Network Repository Function (NRF), are containerized as Docker containers. Each slice has dedicated SMF and UPF instances, while other functions are shared. For resource allocation in the core, we only on the UPF, as it operates in the data plane.

\smallskip
\noindent \textbf{Transport.} The transport network employs a software-defined VXLAN overlay using OvS \cite{ovs} on the underlying physical network. This allows the traffic between the RAN and the core to be routed through the transport network switch.

\subsection{Management and Control}
\noindent \textbf{MANO.} Kubernetes v1.29 is used for orchestrating and managing the 5G VNFs as lightweight containers. The Kubernetes API facilitates the placement of VNFs across distributed nodes and the creation of network slices with desired topologies. We use Linux cgroups to dynamically adjust the CPU resources for these network functions.

\smallskip
\noindent \textbf{SDN Controller.} The ONOS SDN controller \cite{onos} is used to manage network flow routing within slices. It interfaces with OvS switches in the VXLAN overlay to direct slice traffic through OvS queues at predefined rates, enabling efficient bandwidth slicing.

\smallskip
\subsection{Dataset Collection} \label{sec:dataset} To create the dataset for training VNF and slice models, Poisson-distributed traffic is injected into the UPF, Backhaul transport (OvS), and RAN VNFs. Since open-source per-VNF implementations for the RAN are unavailable, the entire RAN is treated as a single VNF. However, the proposed approach can be extended to per-VNF modeling for the RAN once such implementations become available. Traffic is generated at rates ranging from 1 Mbps to 35 Mbps in 5 Mbps increments. Time synchronization across machines is ensured using NTP \cite{ntpdate}. For each traffic profile, VNF resources are varied (\textit{CPU} for the RAN and Core, and \textit{Bandwidth} for the Backhaul transport), and the resulting output traffic is recorded as PCAP files. Data is collected over 60 seconds for each combination of traffic and resource allocation. The captured PCAP files are then pre-processed to extract flow-level feature vectors, which serve as input and output for the VNF and slice models. The dataset is available at \cite{dataset}, along with a simple example in \cite{workshop} that demonstrates how to construct an E2E slice model from VNF models and perform resource optimization.
\section{Experimental Results}\label{sec:results}
We provide three sets of experiments in this section. First, in Section~\ref{sec:exp-net-model}, we evaluate the DNN-based slice model built using real-traces from our testbed. 
Second, we illustrate the performance of our gradient-based resource allocation algorithm (i.e., Algorithm~\ref{algo:res_alloc}) in Section~\ref{sec:exp-resource-alloation}. Finally, we demonstrate the performance of admission control and resource allocation algorithm (i.e., Algorithm~\ref{alg:ota}) in Section~\ref{sec:exp-ac}.  We only evaluate the aspects of the slice model, and the resource allocation algorithms that are relevant to the problem of online slice admission control. For detailed evaluation and comparison with the state-of-the-art, we refer to our previous works \cite{vnetrunner, microopt}. 

\subsection{Slice Model}
\label{sec:exp-net-model}

\begin{figure*}[ht!]
\captionsetup{justification=centering}
\centering
\begin{minipage}{0.75\textwidth}
    \centering
    \subfloat[RAN loss curve]{\includegraphics[width=0.32\linewidth]{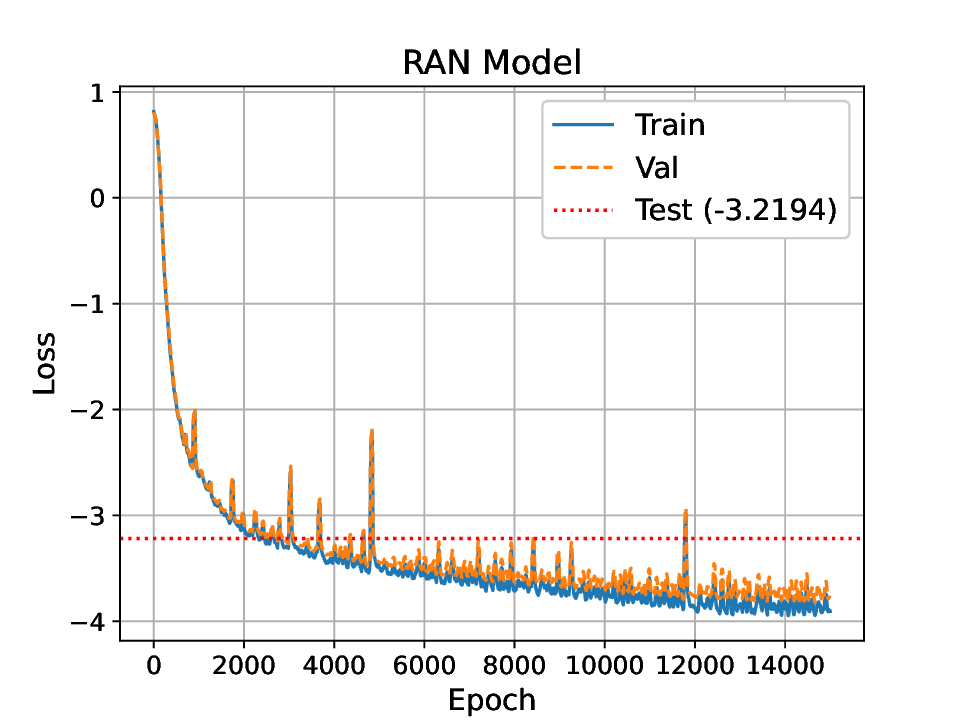}\label{fig:ran_loss}}%
    \hfill
    \subfloat[Backhaul (OvS) loss curve]{\includegraphics[width=0.32\linewidth]{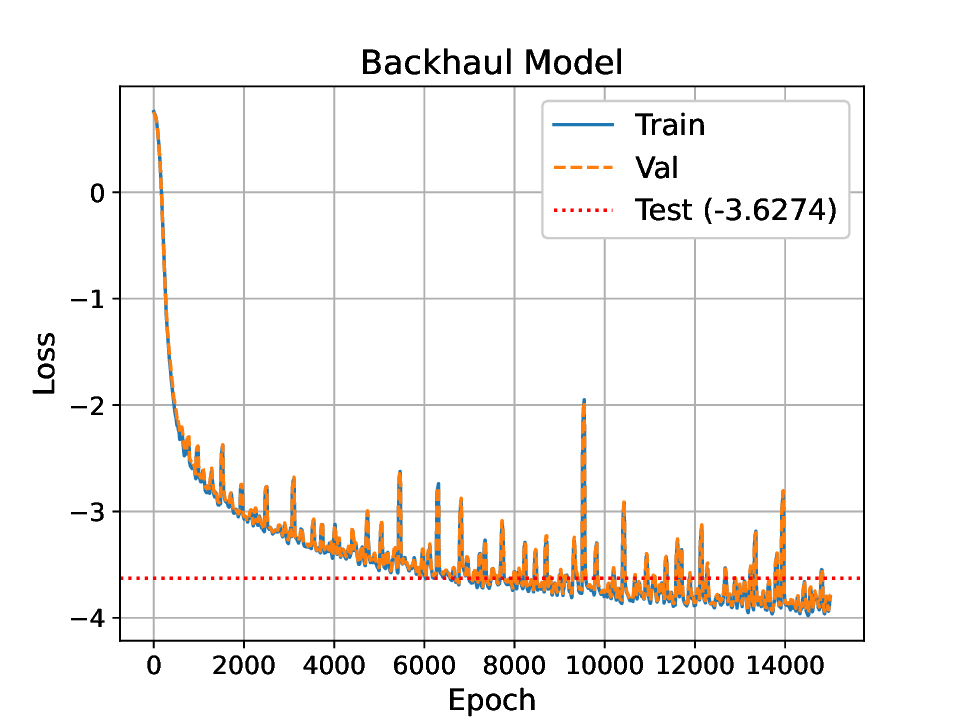}\label{fig:ovs_loss}}%
    \hfill
    \subfloat[Core (UPF) loss curve]{\includegraphics[width=0.32\linewidth]{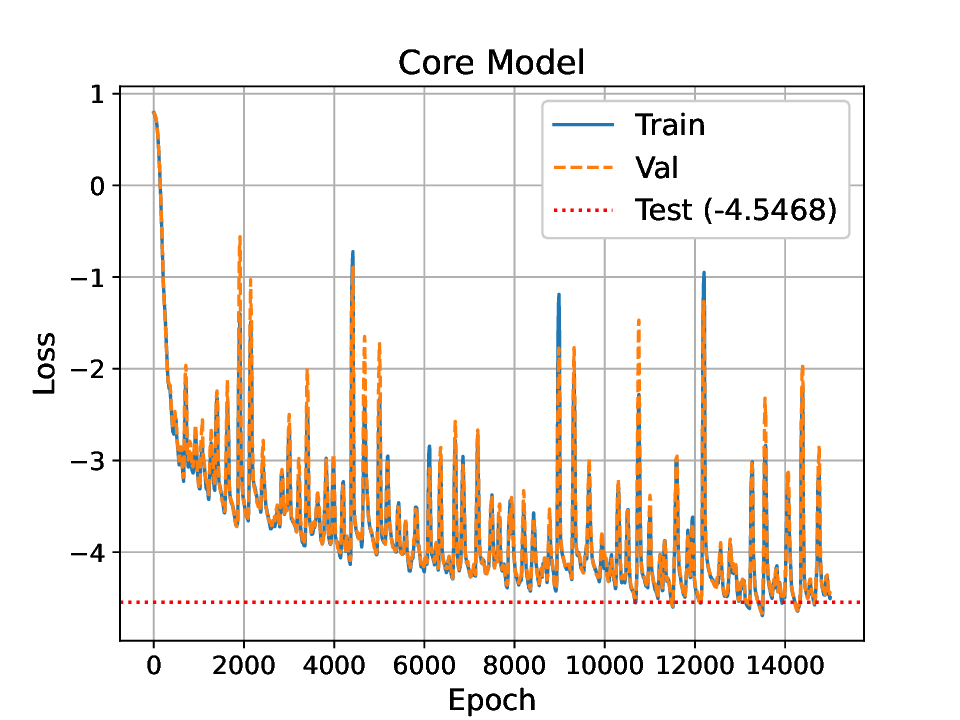}\label{fig:upf_loss}}%
    \caption{VNF model loss curves %(cf. \eqref{eq:lqos})
    }
    \label{fig:vnf_loss_curves}
\end{minipage}%
% \hspace{0.2cm}
\begin{minipage}{0.24\textwidth}
\vspace{0.6cm}
    \centering
    \includegraphics[width=\linewidth]{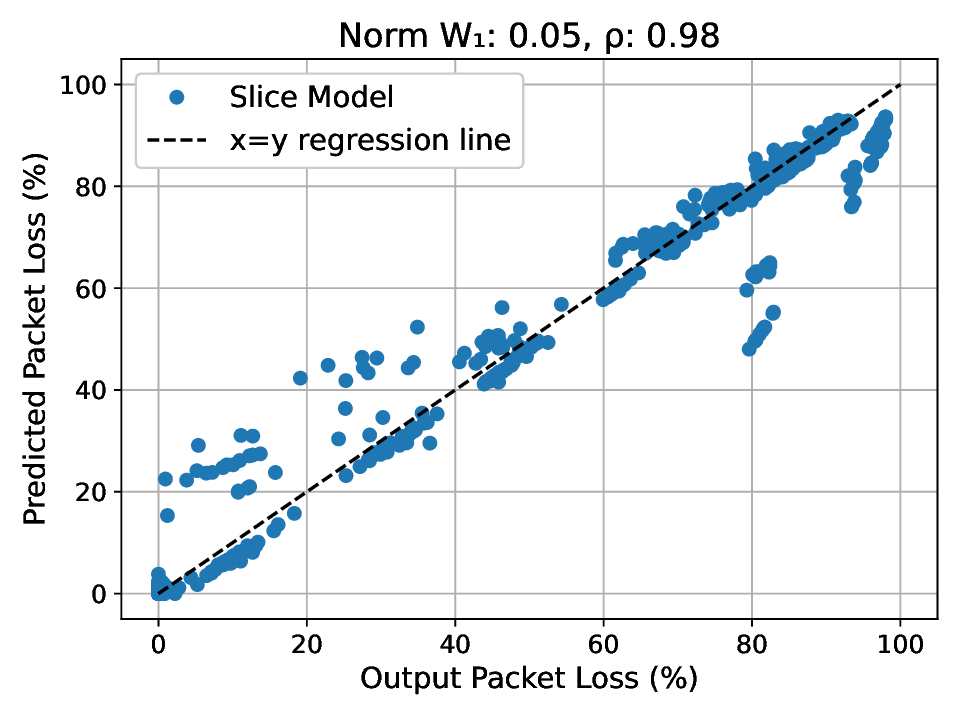}
    \caption{Slice model regression plot %(\eqref{eq:qos_sample}).
    }
    \label{fig:reg_plot}
\end{minipage}
\end{figure*}

\begin{figure*}[ht!]
\captionsetup{justification=centering}

\subfloat[Mean resource allocation]{
    \includegraphics[width=0.24\linewidth]{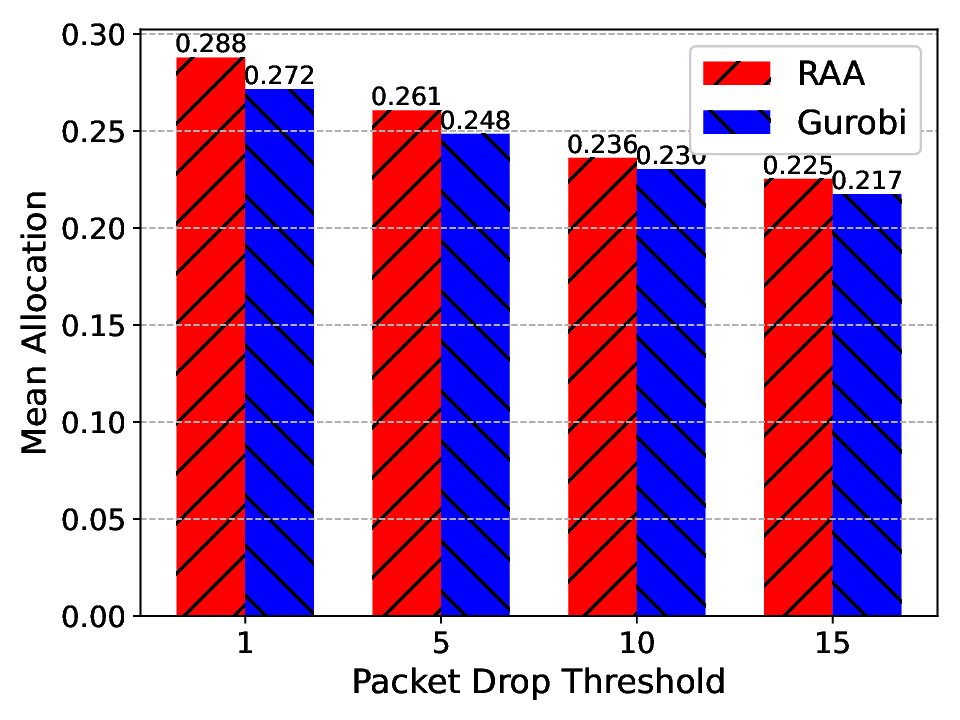}
    \label{fig:mean_res_alloc}
}
\subfloat[Mean runtime]{
    \includegraphics[width=0.24\linewidth]{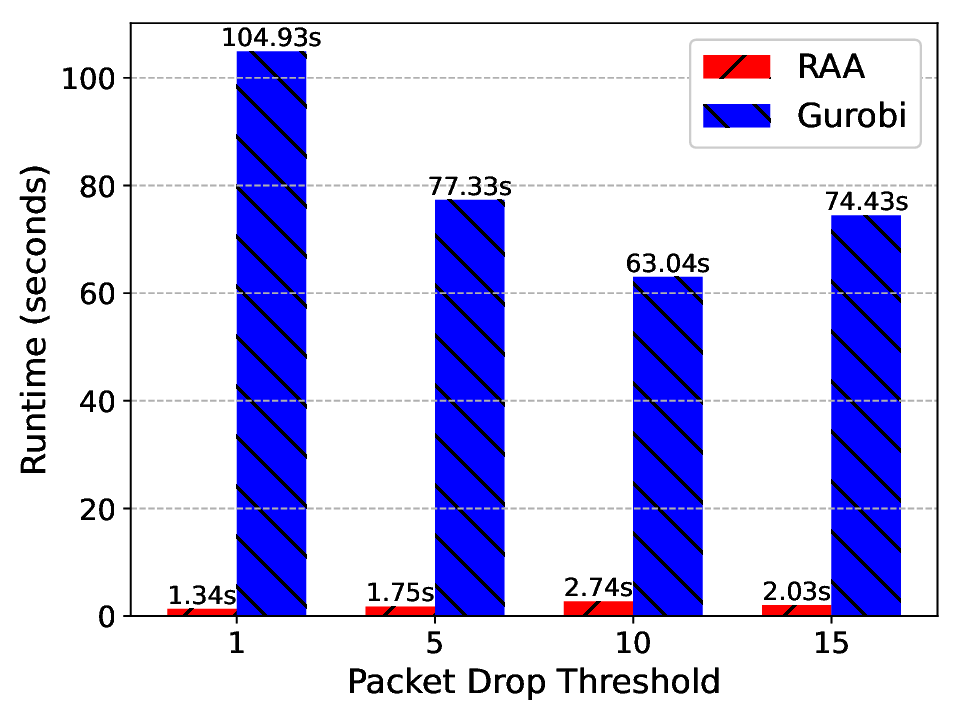}
    \label{fig:mean_runtime}
}
\subfloat[Mean runtime across models]{
    \includegraphics[width=0.24\linewidth]{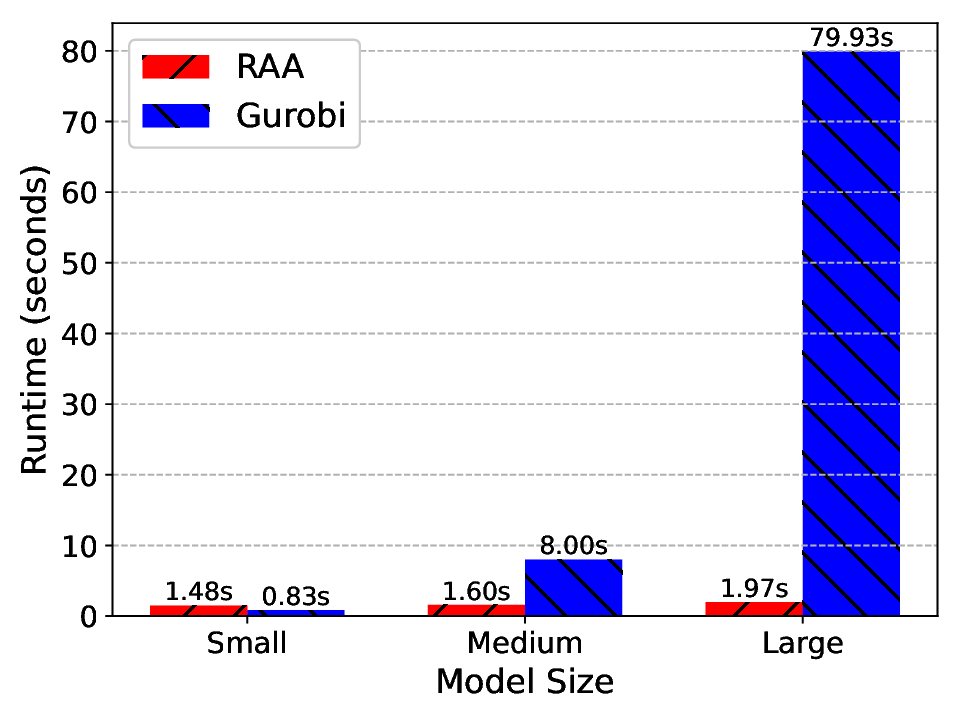}
    \label{fig:runtime_trend}
}
\subfloat[Approximation ratio $\sigma$, $\sigma_{ub}$]{
    \includegraphics[width=0.24\linewidth]{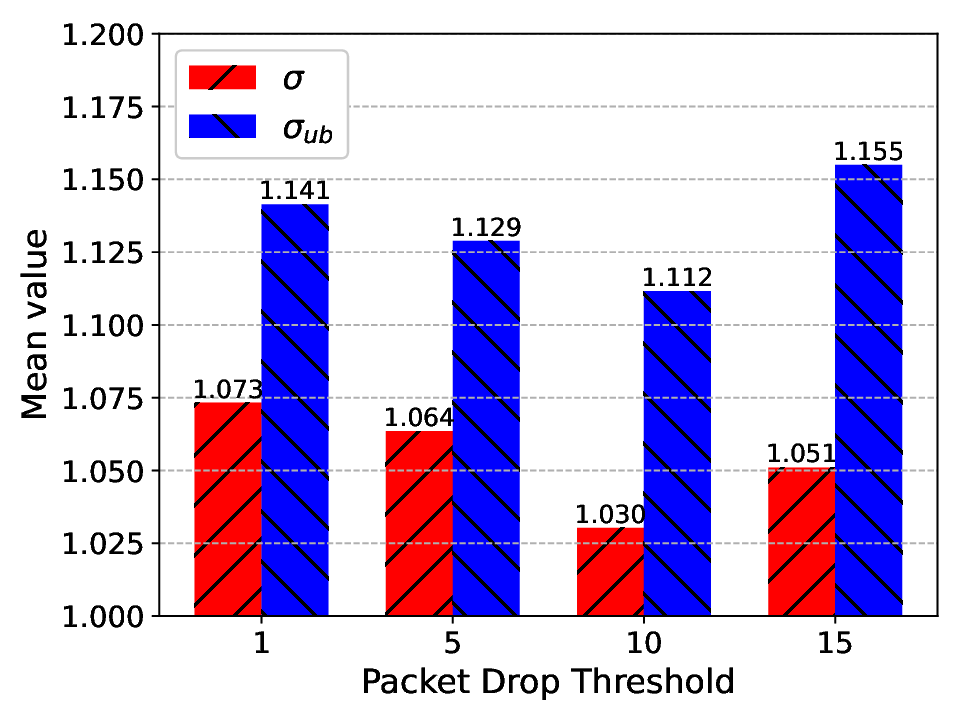}
    \label{fig:sigma}
}
  \caption{Optimality and runtime comparison} 
  \label{fig:res_alloc_stats}
\end{figure*}

A high-level overview of the VNF models is shown in \fig{fig:vnf_model}. For the shared layers, we use two hidden layers with [256, 128] nodes and Rectified Linear Unit (ReLU) activation. The mean and standard deviation branches each have one hidden layer each with 64 nodes. The mean output uses a linear activation function, while the standard deviation output employs the Softplus activation function to ensure non-negativity. Note that the model details, such as the type and number of layers, activation functions, and nodes per layer, may vary for different datasets. For training and inference, we use NVIDIA RTX 4090.

We divide the QoS dataset (cf. \sect{sec:dataset}) into training and validation sets. Additionally, we perform the procedure outlined in \sect{sec:dataset} to gather a test set consisting entirely of off-grid points, i.e., input combinations not present in the training or the validation sets. Subsequently, we train the model for 15,000 epochs with a learning rate of $10^{-4}$. \fig{fig:vnf_loss_curves} shows the negative log probability loss (i.e., $L_{QoS}$ in \eqref{eq:lqos}) as the VNF models train. We can observe that the validation error does not deviate from the training error, which shows that the model is not overfitting to the training data. Once trained, the RAN, Backhaul and Core models achieves a test loss of -3.21, -3.62 and -4.54, respectively. The Core model achieves the lowest training, validation, and test losses because the UPF is significantly simpler compared to the RAN and Core models. This is because the UPF primarily handles packet forwarding, which demands minimal CPU resources. Consequently, its performance remains relatively stable regardless of variations in CPU resource allocation.

Once the VNF models have been trained, we combine them as discussed in Section \ref{sec:slice_model} to compose an E2E slice model. \fig{fig:reg_plot} shows the regression plots capturing the correlation between ground-truth and the predictions. From the figure, we can see that most of the points lie on the $x=y$ regression line. The predictions achieve a Pearson correlation $\rho$ of 0.98 to the ground-truth, showing a high degree of alignment between the prediction and the ground-truth. However, a small number of points deviate from the regression line resulting in a normalized Wasserstein distance\footnote{\cite{yang2022deepqueuenet} defines normalized Wasserstein distance ($W_1$) as: Norm $W_1$ = $W_1(labels, prediction)/W_1([0]*labels, labels)$.} of 0.05. The model can be improved using a larger dataset and more complex ML models, however, a simple feed-forward network suffices for our %small-scale 
testbed.

\subsection{Resource Allocation}
\label{sec:exp-resource-alloation}

In this section, we evaluate the optimality and runtime performance of \microopt. To achieve this, we leverage the recently released Gurobi Machine Learning library \cite{gurobi_ml}, which integrates ML models as constraints within optimization problems. However, the library currently has significant limitations, as it supports only a small selection of layers and activation functions \cite{gurobi_ml_support}. Consequently, for this evaluation, we employ a scalar version of our slice model, where the model predicts a single scalar value rather than an egress traffic feature distribution. While this approach is not ideal---since scalar predictions may result in higher-than-expected QoS degradation \eqref{beta} by failing to account for the full QoS distribution \cite{microopt}---it does provide a basis for direct comparison between \microopt and Gurobi. 

To integrate the E2E slice model within Gurobi, we introduce a constraint, i.e., the input variables to a VNF model must match the output of the upstream VNF model. Additionally, normalization and denormalization of the variables are required as intermediate steps between VNFs. This approach enables the construction of a slice model using pre-trained VNF models. Once the slice model has been integrated with Gurobi, we solve the \csp for several QoS thresholds (i.e., 1\% to 20\% Packet Loss\footnote{By strict definition in Section \ref{sec:prob_statement}, QoS should be defined as \textit{$\text{100}\% -$ Packet Loss}, where higher values indicate better performance. However, for simplicity and without loss of generality, we use packet loss (where lower is better) as the QoS metric for the rest of the manuscript.}) as well as different slice traffic, and compare the performance against \microopt. For this evaluation, we set the resource prices $p_m^{(i)}$ to 1. Therefore, the mean resource allocation is the same as the cost $c_i$, which is the minimization objective in \csp.

\fig{fig:mean_res_alloc} and \fig{fig:mean_runtime} show the mean resource allocation and the mean runtime for the \microopt and Gurobi solutions, respectively. In \fig{fig:mean_res_alloc}, we can see that as the packet drop threshold increases, both the solutions lead to a decrease in the resource allocation required to meet the QoS threshold. However, in all cases,  Gurobi leads to a more optimal solution compared to \microopt. Specifically, \microopt leads to a 5.45\% higher mean resource allocation across the different scenarios. However, in \fig{fig:mean_runtime}, we can see that \microopt requires significantly smaller runtime compared to Gurobi. Across the different scenarios, \microopt leads to more than 40 times faster runtime compared to the Gurobi solution. This highlights the advantage of leveraging gradient information during optimization.

The runtime of Gurobi may depend on the size of the VNF model. Let layer $l$ in the model have $n_{\textit{in},l}$ and $n_{\textit{out},l}$ as the number of inputs and outputs, respectively. Each linear layer with a ReLU activation function results in $n_{\textit{out},l}$ linear constraints and $n_{\textit{out},l}$ general constraints. Consequently, a neural network with $L$ layers results in a total of $\sum_{l=1}^{L-1} 2 \cdot n_{\textit{out},l} + n_{\textit{out},L}$ constraints. For instance, our VNF model with hidden layers comprising [256, 128, 64] nodes and 6 output nodes leads to 454 linear constraints and 448 general constraints. However, the runtime of the Gurobi solution may not always increase with a higher number of constraints. To test the solution time across various VNF model sizes, we define two additional VNF models. We refer to the default model with [256, 128, 64] hidden nodes as `large', and define the new models with [128, 64, 32] and [64, 32, 16] nodes as `medium' and `small' models, respectively. \fig{fig:runtime_trend} shows the mean runtime obtained by Gurobi and \microopt across the three model sizes. We can see that for small model, Gurobi is able to achieve slightly smaller runtime compared to \microopt. However, as the model size increases, the Gurobi runtime increases considerably, and surpasses the \microopt runtime for the medium and large models.

The above analysis highlights the advantage of using \microopt for resource allocation. It can support arbitrarily complex neural network models, such as 
Transformers, Long Short-term Memory \cite{yang2022deepqueuenet}, and GNNs \cite{ferriol2023routenet} through auto-differentiation frameworks (e.g.,  Pytorch \cite{torch}). Additionally, \microopt maintains a small runtime as the model sizes increase. It is worth stating that \microopt may not always lead to a feasible solution, if the initial solution is not feasible. In this case, we suggest using a coarse-grained \texttt{Gridsearch} to find a sub-optimal but feasible solution.

% \subsubsection{Approximation Ratio $\sigma$}
\noindent\textbf{Estimate of Approximation Ratio $\sigma$.} In Section \ref{sec:sigma}, we proposed an algorithm for finding the value for $\sigma_{ub}$ that can be used in place of the approximation ratio $\sigma$, which may not be available during practical scenario. Therefore, we validate the algorithm by calculating the value of $\sigma_{ub}$ and comparing it with $\sigma$. \fig{fig:sigma} compares the value of $\sigma$ and $\sigma_{ub}$ across the various tested scenarios. We can see that across all scenarios, $\sigma_{ub}$ leads to a higher value, \ie a loose upper-bound. Therefore, we can safely use this value in our subsequent online slice admission control algorithm (\opa).

\subsection{Online Slice Admission Control and Resource Allocation}
\label{sec:exp-ac}

\noindent\textbf{Setup.}
Our simulation considers a \sara problem over 200 time slots. As outlined in Section \ref{sec:prob_statement}, each SR $i$ is defined by
% $I_i:= \{v_i, \calt_i, \boldsymbol{\xi}_i, \boldsymbol{\beta}_i\}$, i.e.,
its value $v_i$, its service period $\calt_i$, a feature vector $\boldsymbol{\xi}_i$ (including the QoS requirement and peak traffic distribution), and a QoS degradation threshold $\beta_i$. In the simulation, five SPs submit SRs to the InP at a Poisson-distributed arrival rate of 5. The service period for these slices follows an exponential distribution with a parameter of 4. We define the maximum stay duration, $K$, as the 99th percentile value, which is 18.42. We assume that the slice traffic ($\Lambda_i$) remains constant, sampled from a uniform distribution of 1-7 users, where each user generates 5Mbps of ingress traffic. The QoS requirement, denoted as $q^i_{thresh}$, is randomly chosen from the set of $\{1\%, 5\%, 10\%, 15\%, 20\%\}$ Packet losses. The value $v_i$ offered by each SR is sampled from a uniform distribution $\mathcal{U}[1, 2]$. 
% \bo{Is it value or value density. Are $1$ and $2$ corresponding to $L$ and $U$?}. 
Given the relative scarcity of bandwidth resources compared to compute resources in mobile networks, we set the normalized resource capacity for CPU and bandwidth as $5$ and $10$, respectively. The above mentioned simulation parameters are chosen based on a number of related works \cite{sulaiman2022coordinated,TNET-1,TNSM_GNN}.
Based on the network model and the simulation parameters delineated above, we have calculated key parameters $L, U, V,$ and $\alpha$ as 0.045, 188.5, 6.0, and 27.43, respectively. 

We employ the commercial solver Gurobi to obtain the offline optimal solution. However, given the limited layer and activation function supported by Gurobi Machine Learning library, integrating the full data-driven slice model into Gurobi, and performing SAC presents a non-trivial challenge. To circumvent this problem, we deploy the resource allocation algorithm \microopt as a preprocessing step. We repeatedly pre-solve the \csp in~\eqref{p:csp} for a given pseudo-price, generating a set of feasible resource allocations. 
Since the solution only depends on the relative price of the resources, we pre-solve the resource allocation problem for various combinations of relative resource prices, for all different slice request possibilities. This enables us to reframe the offline resource allocation problem as a problem to select the cost minimized feasible allocation, and this selection problem can be seamlessly solved by Gurobi. To ensure a fair and consistent comparison, we also constrain the solution space of the \csp in the online algorithm ($\opa$) to the same set of feasible allocations.

\begin{figure}[t!]
 \centering
  \includegraphics[width=0.6\linewidth]{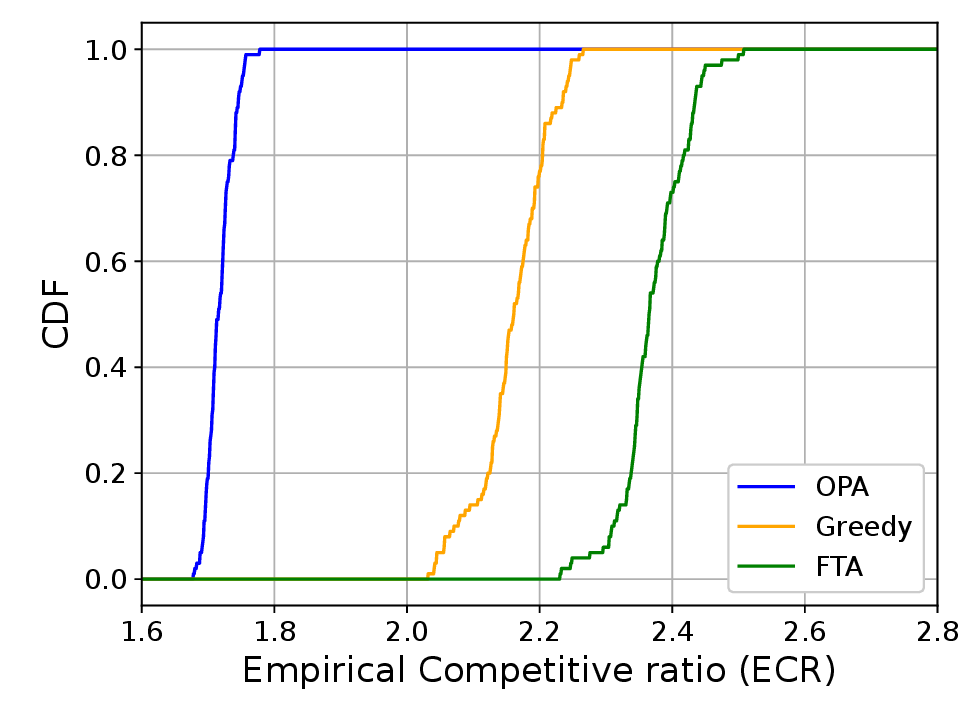}
    \caption{CDF of empirical competitive ratios}
    \label{fig:cr_cdf}
\end{figure}

\begin{figure*}[ht!]
\captionsetup{justification=centering}

\subfloat[Empirical competitive ratio]{
    \includegraphics[width=0.24\linewidth]{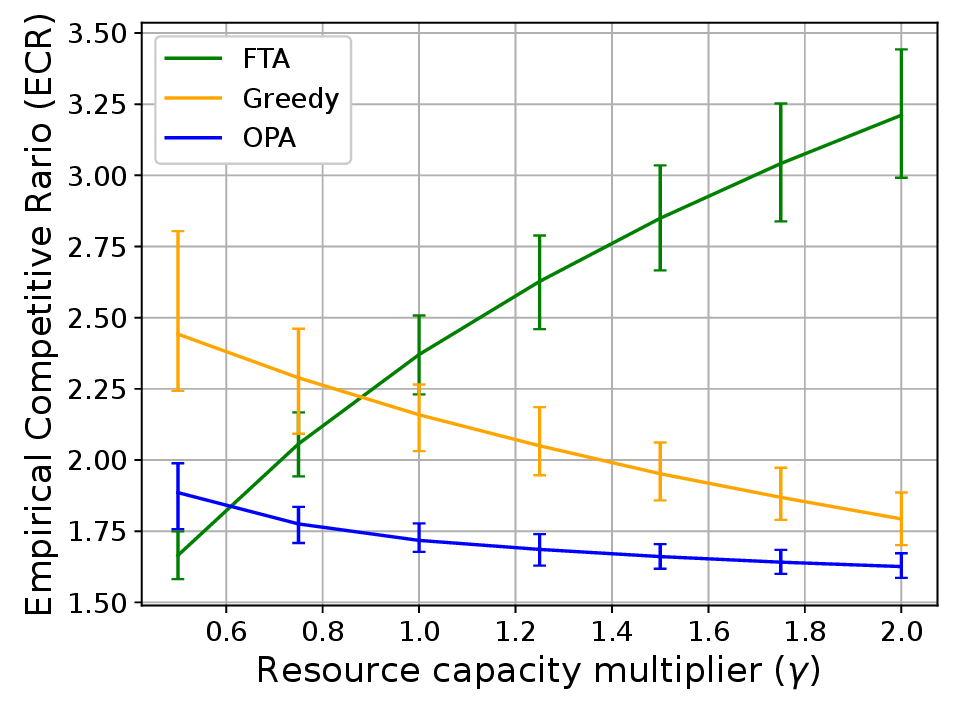}
    \label{fig:mean_cr}
}
\subfloat[Acceptance ratio]{
    \includegraphics[width=0.24\linewidth]{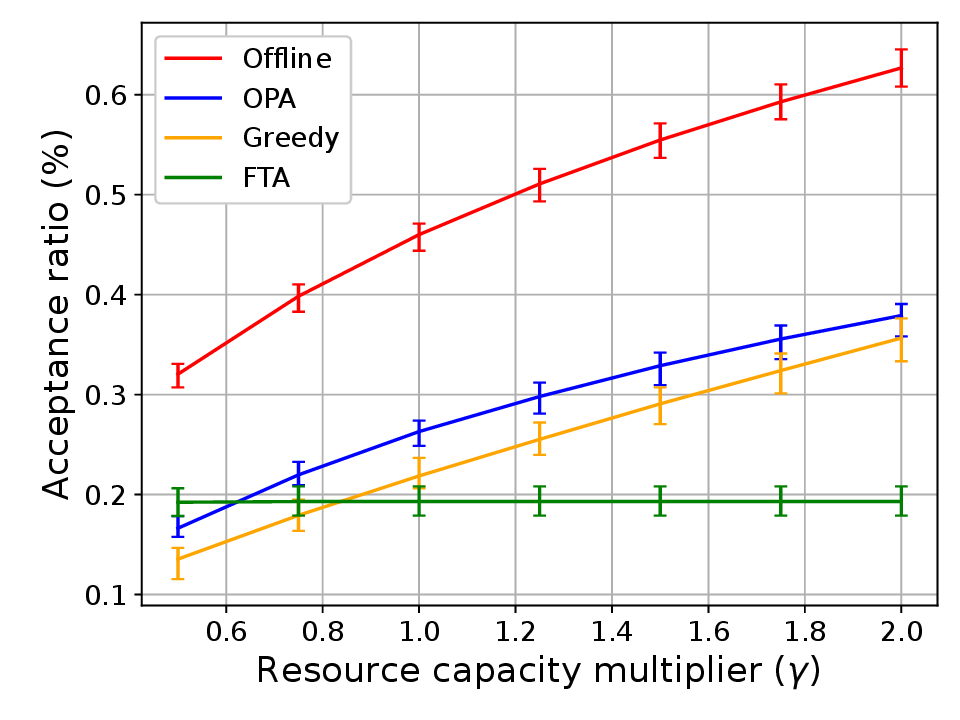}
    \label{fig:mean_ar}
}
\subfloat[Bandwidth resource utilization]{
    \includegraphics[width=0.24\linewidth]{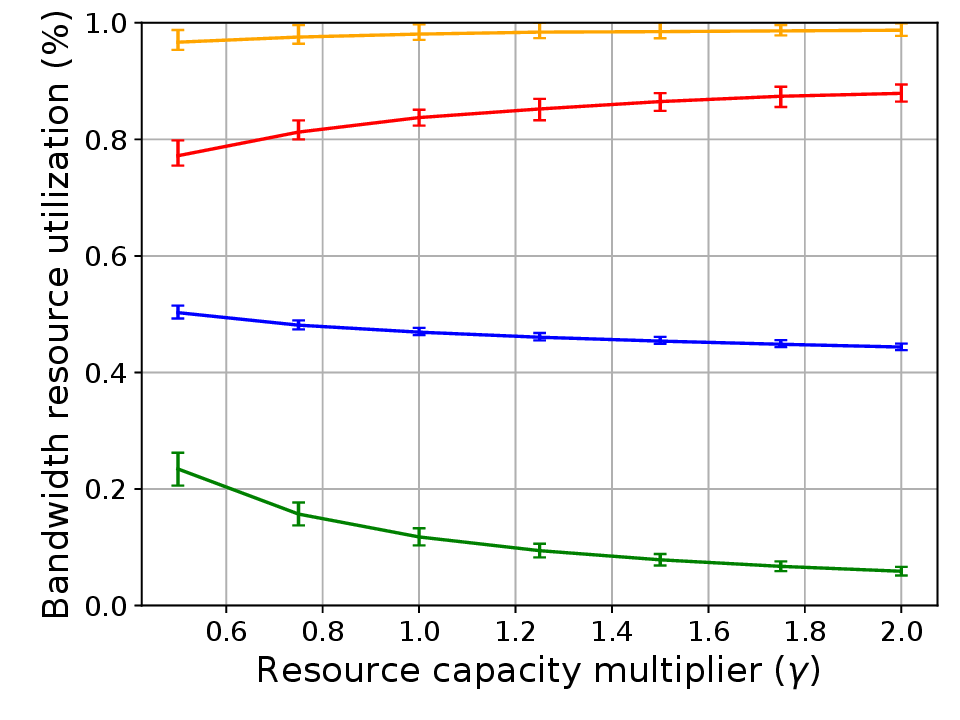}
    \label{fig:res_1}
}
\subfloat[Compute resource utilization]{
    \includegraphics[width=0.24\linewidth]{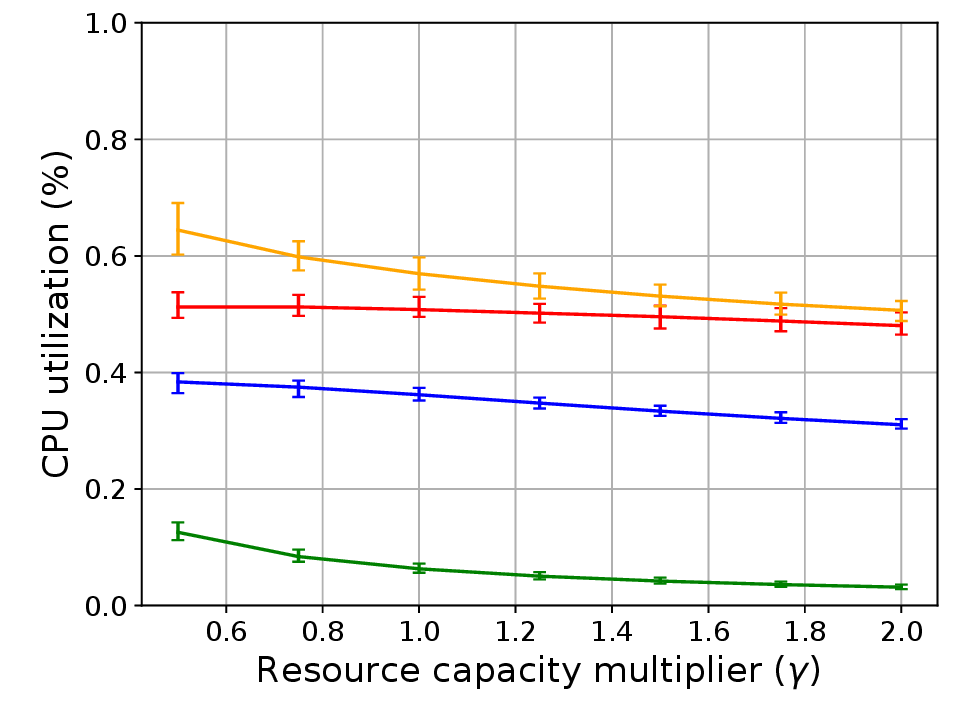}
    \label{fig:res_2}
}
  \caption{Comparison of online algorithms for SAC with varying resource capacities} 
  \label{fig:ac_stats}
\end{figure*}

\noindent\textbf{Comparison Approaches.} We compare our proposed \opa (in Algorithm~\ref{alg:ota}) with two heuristic AC policies describe below. However, it is important to note that we do not compare our approach against data-driven SAC methods, such as RL or multi-armed bandit algorithms. These methods require prior knowledge of the slice request distribution (e.g., inter-arrival times, service periods) for training. In contrast, our work addresses the online version of the SAC problem, which assumes no prior knowledge of future slice requests. Therefore, we limit our comparison to the Greedy and \fta approaches, which are commonly used in online settings~\cite{lechowicz2024online,sun2020competitive}.

\textit{Greedy.} 
This method accepts all incoming SRs and opts for the resource allocation strategy that minimizes the sum of resources used, while adhering to the resource capacity constraints. The Greedy approach tends to be effective in scenarios where the system load is low and the risk of resource bottlenecks is small. However, its performance may deteriorate under higher load conditions. In such cases, 
aggressively admitting slices may accept slices that arrive early but are of low values, while rejecting high-value slices that arrive later. 

\textit{Fixed-Threshold Algorithm (\emph{\fta}).} In contrast to our proposed \opa, which dynamically adjusts resource prices based on their utilization, \fta heuristically sets a fixed threshold price for SAC. The fixed price is set to strike a balance between greediness and conservativeness in accepting SRs. In particular, we adopt the fixed price $\sqrt{LU}$ that has been suggested for the online selection problem~\cite{el2001optimal}. This comparison allows  to gauge the effectiveness of \opa's dynamic pricing strategy against a fixed, yet well-established heuristic for resource allocation.

\noindent\textbf{Evaluation.}
We evaluate the performance of online algorithms for \osac based on empirical competitive ratios (ECRs). 
Based on the experimental setup, we generate a total of $100$ instances. For each instance $\cali$, we evaluate the objective value of an online algorithm under this instance and the offline optimal objective $\opt(\cali)$. The ECR of an instance $\cali$ is $\opt(\cali)/\alg(\cali)$. We demonstrate the performance of online algorithms using the statistics of ECRs over multiple instances.  
% To access the performance of the proposed admission control algorithm (\opa), we generate slice requests using the previously discussed simulation parameters. 
\fig{fig:cr_cdf} shows the cumulative density function (CDF) of ECRs achieved by the three online algorithms. 
Compared to Greedy and \fta, our proposed \opa 
% achieves the minimum average CR followed by the Greedy and \fta approaches, respectively. 
% From \fig{fig:cr_cdf}, it can also be seen that \opa 
not only achieves the minimum mean ECR but also excels in worst-case scenarios.
% as evidenced by its minimum maximum ECR. 
Specifically, when examining the mean ECR across all approaches, \opa achieves a mean ECR of $1.72$, which is notably lower than that of the \fta ($2.37$) and Greedy ($2.15$) approaches. In the context of worst-case performance, the proposed approach shows a maximum ECR of $1.77$, substantially outperforming both the \fta and Greedy approaches, which have maximum ECRs of $2.50$ and $2.26$, respectively. This underscores the proposed approach's ability to maintain a more stable and predictable performance even under challenging conditions.

To evaluate the performance of online algorithms under different levels of resource scarcity, we introduce a resource capacity multiplier $\gamma$ to adjust the available resources by scaling the base resource capacity $C_m$ by $\gamma$, with a larger $\gamma$ representing a smaller load. \fig{fig:mean_cr} and~\fig{fig:mean_ar} depict the ECRs and acceptance ratios (i.e., the proportion of admitted SRs among all SRs in an instance) across varying values of $\gamma$.
From \fig{fig:mean_cr}, it can be observed that \opa outperforms both Greedy and \fta approaches in terms of mean and maximum ECRs over the varying loads, showcasing the superior performance of \opa in both average and worst-case scenarios. 
We can observe that under small load conditions (\ie larger $\gamma$), both the proposed and Greedy approaches exhibit similar acceptance ratios. In this scenario, the Greedy approach's performance is comparable to \opa, which can be attributed to the abundance of resources. This abundance allows for the acceptance of low-value SRs without compromising the resource availability for potentially higher-value future SRs. Conversely, the \fta approach demonstrates a markedly conservative stance, resulting in significantly suboptimal performance when the load is small. Its conservative nature leads to missed opportunities for gains from low-value SRs. However, as the system load intensifies and the optimal acceptance ratio converges towards $50\%$, the Greedy approach's effectiveness diminishes. In this mid-load range, our proposed approach achieves the lowest competitive ratio by effectively navigating between opportunistic and cautious resource pricing. The scenario shifts further with increasing load, where the optimal acceptance ratio approaches $35\%$. At this juncture, the conservative nature of \fta becomes advantageous. Its predisposition to reject low-value SRs increases the likelihood of admitting future high-value SRs, achieving the best ECR under larger load conditions. Across all the different scenarios, \opa, \fta and greedy approaches achieve a mean ECR of $1.71$, $2.54$ and $2.07$, and a max ECR of $1.98$, $3.44$, and $2.80$, respectively. This shows that \opa achieves a mean improvement of $32.7\%$ and $17.6\%$ and a worst-case (max ECR) improvement of $42.23\%$ and $29.1\%$ over \fta and greedy approaches, respectively.

From \fig{fig:mean_ar}, we can see that the Greedy approach achieves an unexpectedly low acceptance ratio. To explain this behavior, we show the resource utilization of the different approaches under different load conditions in \fig{fig:res_1} and \fig{fig:res_2}. It can be observed that the Greedy approach indeed attains a high resource utilization. However, by accepting SRs that arrive early and are with high resource demand, the Greedy approach quickly uses up resources and creates a bottleneck for future SRs with low resource demand. As a result, the acceptance ratio of Greedy is low compared to other approaches. The resource utilization for price-based algorithms (\ie \opa and \fta), shows the expected behavior, i.e., opportunistically reserving resources for SRs with a high value and a low resource requirement.

\noindent\textbf{Robustness and Explainability.}
Compared to \opa, both \fta and Greedy are vulnerable to worst-case scenarios. Particularly, for \fta, if SRs have value densities just below the fixed threshold, \fta can potentially achieve unbounded worst-case ECR, approaching infinity. As for the greedy approach, SRs with the lower bound of value density may saturate one resource early, followed by SRs offering the upper bound of value density for all resources, resulting in a worst-case ECR of approximately $M U/L$. In contrast, our proposed algorithm \opa dynamically adjusts resource prices based on current utilization, ensuring robustness against such adversarial scenarios. As proven in Theorem \ref{thm:osac}, no set of SR arrivals can lead to a worst-case ECR exceeding $\frac{(\sigma + 1)\alpha}{2}$. This makes \opa robust and trustworthy for practical applications.

\begin{figure}[h!]
 \centering
  \includegraphics[width=0.6\linewidth]{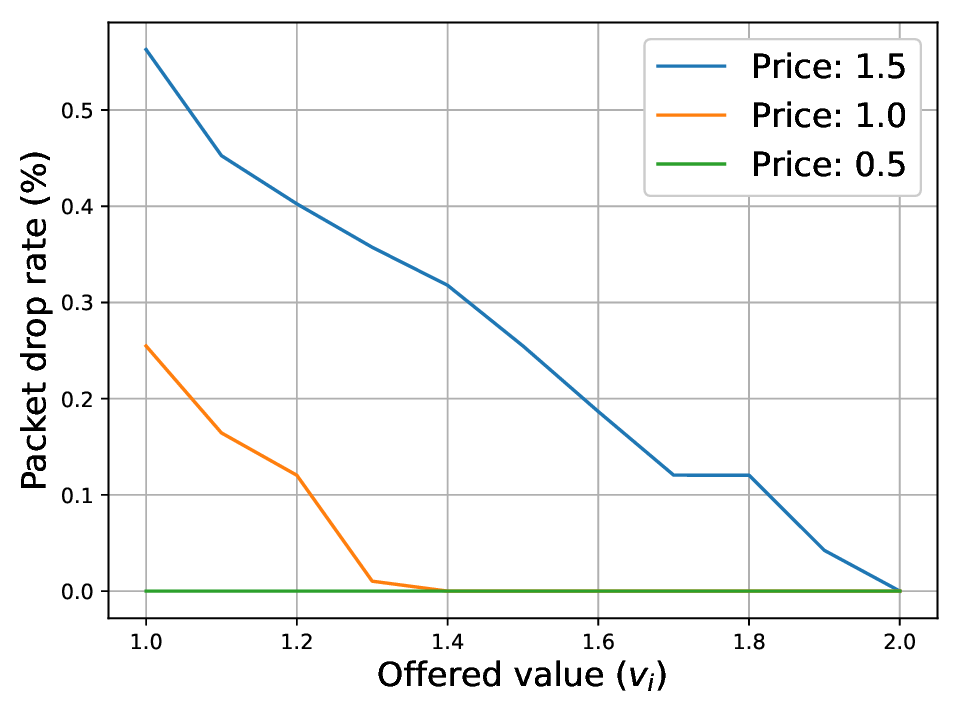}
    \caption{Maximum offerable QoS at different resource prices, and offered revenues}
    \label{fig:explainability}
\end{figure}

Several studies (e.g., \cite{sulaiman2022coordinated,Ghina_2020,TNSM_GNN}) in the literature leverage RL for SAC. These black-box approaches typically provide only an admission control decision for the currently arriving slice request, without any insight or reasoning behind the decision. In contrast, the pricing-based approach presented here allows the InP to precisely determine the difference between the resource allocation cost incurred to accept a slice request (i.e., using \microopt), and the offered revenue. This enhances explainability by making admission control decisions transparently tied to the cost and revenue. Additionally, \eqref{p:csp} can be reformulated to maximize the mean QoS given the offered value as follows: 

\begin{align}
\max_{\by_{i}} \quad& \mathbb{E}(q_{i}(n_i, \by_{i}))\nonumber\\
\label{eq:qsp}
    {\rm s.t.}\quad 
    & \sum\nolimits_{t\in\calt_i}\sum\nolimits_{m\in\calm} p_{m,t}^{(i)} y_{i,m}^t \leq v_i \\
    & 0 \le y_{i,m}^t \le R_{m}, \forall m\in\calm, t\in\calt_i\nonumber.
    % &y_{i,m} \ge 0, \forall i\in\cals, m\in\calm.
\end{align}

This constrained optimization problem can be addressed using the algorithm proposed in Section \ref{sec:opt}. To validate this, we solve \eqref{eq:qsp} under varying resource prices and offered revenues for a traffic rate of 7 users/s. As shown in \fig{fig:explainability}, higher resource prices (i.e., reflecting a more congested substrate network) necessitate a higher offered value to satisfy a given QoS threshold. This enables the InP to utilize real-time resource prices provided by the \opa to estimate the best achievable QoS for a given SR. Additionally, this information can be leveraged by the SP to adjust SR parameters (e.g., QoS requirements, offered revenue), which allows for a transparent and explainable SAC process.
\section{Conclusion} \label{sec:conclusion}

We introduced a novel data-driven framework for online SAC and resource allocation in 5G and beyond networks. Our resource allocation algorithm employs a DNN with the reparameterization trick to create a differentiable network model. This model allows for effective primal-dual optimization to minimize resource allocation cost, while adhering to QoS constraints. The AC algorithm integrates the resource allocation algorithm by dynamically adjusting the pseudo-prices for different resources, representing their scarcity. By strategically setting these prices, we prove that our AC algorithm maintains a bounded competitive ratio. Empirical results show that our network model effectively generalizes to test data, evidenced by minimal negative log probability loss. Furthermore, our resource allocation algorithm consistently exhibits a negligible optimality gap which having a significantly faster runtime compared to Gurobi. Finally, we compared our AC solution (\opa) with both greedy and fixed threshold-based methods. Our findings reveal that \opa consistently outperforms these methods, achieving the lowest mean ECR across various load scenarios and leads to up to $42\%$ improvement in worst-case ECR in the tested scenario.

One limitation of our proposed AC algorithm lies in its inherent conservativeness for attaining worst-case performance guarantees. To enhance its average-case performance and adaptability, a promising future direction is to fine-tune the value of $\alpha$ within the current pricing function, or directly learn the pricing function using AI techniques based on historical data. Another possible direction is to investigate integrating resource scaling into the framework to dynamically scale the resources of a slice once it has been accepted. In addition, we will also consider enhancing the network model to include additional KPIs, such as jitter and latency.

\section*{Acknowledgement}
This work was supported by Rogers Communications Canada Inc.

\small
\bibliographystyle{IEEEtranN}
\bibliography{IEEEabrv,main}

\end{document}